\documentclass[sigconf,nonacm]{acmart}

\newcommand\vldbdoi{XX.XX/XXX.XX}
\newcommand\vldbpages{XXX-XXX}
\newcommand\vldbvolume{14}
\newcommand\vldbissue{1}
\newcommand\vldbyear{2020}
\newcommand\vldbavailabilityurl{https://github.com/cryptosdu/MESS}
\newcommand\vldbpagestyle{plain}

\usepackage{amsmath,amsthm,mathtools}
\usepackage{bm}
\usepackage{xspace}
\usepackage{xcolor,color,colortbl}
\usepackage{graphicx}
\usepackage{booktabs}
\usepackage{multirow}
\usepackage{adjustbox}
\usepackage{enumitem}
\usepackage{multicol}
\usepackage{float}
\usepackage{placeins}
\usepackage{stfloats}
\usepackage{pifont}
\usepackage{soul}
\usepackage{svg}
\usepackage{etoolbox}
\usepackage{threeparttable}

\usepackage{algorithm}
\usepackage{algpseudocode}
\algrenewcommand{\algorithmicrequire}{\textbf{Input:}}
\algrenewcommand{\algorithmicensure}{\textbf{Output:}}

\usepackage{header/protocolj}
\usepackage{header/game}
\usepackage{arydshln}

\graphicspath{{figures/}}

\usepackage{tikz}
\usepackage{pgfplots}
\usepgfplotslibrary{groupplots}
\pgfplotsset{compat=1.18}

\usepackage{subcaption}
\usepackage[normalem]{ulem}

\hypersetup{
  colorlinks=true,
  citecolor=black,
  urlcolor=black,
  linkcolor=black,
  anchorcolor=black
}
\usepackage{caption}
\usepackage{subcaption}    
\newcommand{\name}{\ensuremath{\textrm{MESS}}\xspace}

\newcommand{\cmark}{\ding{51}}
\newcommand{\xmark}{\ding{55}}

\makeatletter
\@ifundefined{definition}{}{}
\@ifundefined{theorem}{\newtheorem{theorem}{Theorem}}{}
\@ifundefined{lemma}{}{}
\@ifundefined{remark}{}{}
\@ifundefined{proposition}{\newtheorem{proposition}{Proposition}}{}
\@ifundefined{corollary}{\newtheorem{corollary}{Corollary}}{}
\makeatother

\definecolor{PaperBlue}{RGB}{78,121,167}
\definecolor{PaperOrange}{RGB}{242,142,43}
\definecolor{PaperGreen}{RGB}{89,161,79}
\definecolor{PaperRed}{RGB}{225,87,89}
\definecolor{PaperPurple}{RGB}{117,112,179}
\sethlcolor{yellow}

\pgfmathdeclarefunction{erf}{1}{%
  \begingroup
  \pgfmathparse{#1 < 0 ? -1 : 1}%
  \edef\sign{\pgfmathresult}%
  \pgfmathparse{abs(#1)}%
  \edef\x{\pgfmathresult}%
  \pgfmathparse{1/(1+0.3275911*\x)}%
  \edef\t{\pgfmathresult}%
  \pgfmathparse{%
    \sign*(1 - (((((1.061405429*\t - 1.453152027)*\t + 1.421413741)*\t - 0.284496736)*\t + 0.254829592)*\t)*exp(-\x*\x))%
  }%
  \pgfmathsmuggle\pgfmathresult
  \endgroup
}

\theoremstyle{plain}

\numberwithin{theorem}{section}

\begin{document}

\setcounter{page}{1}

\title{\name: Fast and Private Semantic Search on Multi-Graph HNSW}

\author{Haoyu Cui}
\email{cuihaoyu@mail.sdu.edu.cn}
\affiliation{
  \institution{Shandong University}
  \country{China}
}

\author{Zengpeng Li}
\email{lizengpeng@sdu.edu.cn}
\affiliation{
  \institution{Shandong University}
  \country{China}
}

\author{Tien Tuan Anh Dinh}
\email{anh.dinh@deakin.edu.au}
\affiliation{
  \institution{Deakin University}
  \country{Australia}
}

\author{Mei Wang}
\email{wangmeiz@sdu.edu.cn}
\affiliation{
  \institution{Shandong University}
  \country{China}
}

\begin{abstract}
Semantic search systems map data to a high-dimensional vector space and support retrieval of similar data via approximate nearest neighbor search. When the system is hosted by an untrusted cloud provider, there is no privacy for the data or the query. Our goal is to design a system with three properties: privacy, accuracy, and efficiency. Existing works adopt either homomorphic encryption (HE), oblivious RAM (ORAM), or a differential privacy (DP) approach. They fall short of achieving all three properties. 

In this paper, we present \name{}, a system that realizes our goal. It  maps the original vectors into binary codes, applies locality-sensitive hashing (LSH) and randomized response, and constructs a multi-graph Hierarchical Navigable Small World (HNSW)
index over the perturbed codes. \name{} ensures privacy of data, queries, and access patterns. It also ensures search pattern privacy via a two-phase query perturbation mechanism. The multi-graph index mitigates the impact of perturbation on result quality, thereby achieving accuracy. \name{} is efficient because search is performed directly over perturbed codes, without the overhead of homomorphic encryption or ORAM. We provide a formal analysis of the system's privacy and an extensive evaluation of its performance. The results show that \name{} achieves up to $15.08\times$ lower latency than state-of-the-art baselines.
\end{abstract}

\maketitle

\pagestyle{\vldbpagestyle}

\begingroup
\small
\noindent
\raggedright
\textbf{PVLDB Reference Format:}\\
Haoyu Cui, Zengpeng Li, Tien Tuan Anh Dinh, and Mei Wang.
\name: Fast and Private Semantic Search on Multi-Graph HNSW.
PVLDB, \vldbvolume(\vldbissue): \vldbpages, \vldbyear.\\
\href{https://doi.org/\vldbdoi}{doi:\vldbdoi}
\endgroup

\begingroup
\renewcommand\thefootnote{}\footnote{\noindent
This work is licensed under the Creative Commons BY-NC-ND 4.0 International License.
Visit \url{https://creativecommons.org/licenses/by-nc-nd/4.0/} to view a copy of this license.
For any use beyond those covered by this license, obtain permission by emailing
\href{mailto:info@vldb.org}{info@vldb.org}.
Copyright is held by the owner/author(s).
Publication rights licensed to the VLDB Endowment.\\
\raggedright
Proceedings of the VLDB Endowment, Vol. \vldbvolume, No. \vldbissue\ %
ISSN 2150-8097.\\
\href{https://doi.org/\vldbdoi}{doi:\vldbdoi}
}
\addtocounter{footnote}{-1}
\endgroup

\ifdefempty{\vldbavailabilityurl}{}{
\begingroup
\small
\noindent
\raggedright
\textbf{PVLDB Artifact Availability:}\\
The source code and supporting artifacts will be made publicly available upon
acceptance.
\endgroup
}

\section{Introduction}\label{sec:intro}
Semantic search systems support a wide range of applications, from classic information retrieval, image matching~\cite{lowe2004distinctive,sivic2003video}, and recommendation~\cite{sarwar2001item} to pattern recognition~\cite{cover1967nearest,weinberger2009distance} and recent retrieval-augmented generation (RAG) applications. A typical semantic search system maps data and queries to high-dimensional vectors (or embeddings). Its main operation is approximate nearest neighbor (ANN) search, which finds the closest vectors to a given search vector using a distance metric. Recent RAG applications, driven by the success of large language models (LLMs), have led to renewed interest in ANN algorithms, indices, and vector database systems that enable efficient and scalable semantic search~\cite{hu2025hakes, sembench, wang2025leann}. As data and query volumes grow, semantic search systems are migrating to the cloud~\cite{googlevector,amazonsearch,weaviate}. Specifically, cloud-based semantic search services are attractive because they can be efficient, accurate, cost-effective, and scalable. However, they offer no privacy guarantees, as they operate on raw (or plaintext) embeddings and user queries, which can reveal sensitive information about the data and the user. As a consequence, they do not support applications that process sensitive data, such as private search over personal photos and
documents and  personal AI agents~\cite{kaviani2026opal}. 



Our goal is to design a semantic search system with three properties: privacy, accuracy, and efficiency. The first property covers privacy of data, queries, access patterns, and search patterns. The second property means that the search results have high recall, while the last property means low search overhead. The main challenge in achieving this goal is to satisfy all three properties simultaneously. 
%
%
Existing works on private semantic search adopt one of the following three approaches. First, the homomorphic encryption and secure-computation approach, for example, ~\cite{chen2020sanns,Riazi2019SubLinearPPNNS,Liu2025PPANNS,li2025panther}, computes vector distance directly on encrypted vectors. This approach achieves privacy but is inefficient due to the overhead of encrypted computation with high-dimensional vectors. Even when using clustering to reduce the search space, the search cost is still significant because of the large number of candidate vectors. Second, the ORAM approach, for example \cite{zhu2025compass,ahmad2021coeus}, achieves privacy by hiding the vectors being accessed during search. In particular, 
Compass~\cite{zhu2025compass} combines an efficient vector index, ORAM, and product quantization to
protect graph-access patterns~\cite{zhu2025compass}. However, it still fails to achieve efficiency because adaptive graph
traversal requires multiple ORAM operations, increasing communication overhead. Third, the differential privacy (DP) approach, for example \cite{Fernandes2021LSHXDP,Aumuller2020DPJaccard,chen2018differentially}, trades privacy for efficiency by reducing strict privacy to quantified leakage and allowing distance computation on perturbed vectors. However, existing works do not consider any index, thus suffering from high search overhead. They fail to achieve accuracy due to recall degradation caused by vector perturbations. 

We present \name{} that achieves our goal. It addresses the above challenge via a novel combination of differential privacy and a multi-graph index. In particular, it adopts the hierarchical navigable small world (HNSW) graph index for ANN search. It applies locality-sensitive hashing (LSH), followed by an LSH randomized response mechanism, to the original embeddings before building the graph. The DP mechanism ensures privacy of data and access patterns while enabling efficient computations over perturbed vectors. \name{} overcomes the accuracy challenge of the DP-based approach by building and searching over multiple graphs, which reduces rank distortion and candidate set expansion. \name{} introduces a two-phase query-perturbation process that restricts leakage of search patterns. 
\autoref{tab:related} compares \name{} with other works providing private keyword and semantic search. In summary, we make the following contributions.
\begin{itemize}[leftmargin=1.35em,itemsep=0.05em,topsep=0.1em]
\item We propose an efficient HNSW search over perturbed binary embeddings. The client maps data and query vectors to binary codes, then applies an LSH randomized response mechanism to them. The server builds
multiple HNSW graphs over the perturbed codes. It performs a Hamming-distance search on the graphs and returns the results in a single round. 

\item We propose a multi-graph design that addresses search quality degradation caused by perturbations, and a two-stage randomization process to protect repeated queries against averaging and direct linkage attacks. In particular, each query is routed to multiple graph indices and their results are aggregated. Furthermore, the client applies a permanent and an instantaneous randomization to each query. 

\item We provide formal analysis of the privacy bounds for the stored index, query access patterns, and repeated-query search patterns. We evaluate representative cross-shard inference methods to measure practical leakage.

\item We evaluate \name{} on standard datasets and compare it against non-private, homomorphic-encryption, and ORAM baselines. The results show that \name achieves up to \(15.08\times\) lower latency and \(35.28\times\) lower communication overhead than Compass~\cite{zhu2025compass}. On the SIFT100M dataset, \name{} achieves 52.53 ms/query, demonstrating its scalability. 
\end{itemize}

\begin{table*}
\centering
\scriptsize
\caption{Comparison of private keyword and semantic search solutions.}
\label{table:pir_comparison}
\renewcommand{\arraystretch}{1.15}
\setlength{\tabcolsep}{2.5pt}

\resizebox{\textwidth}{!}{%
\begin{tabular}{lccccccccc}
\toprule
\multirow{3}{*}{\textbf{Scheme}}
& \multirow{3}{*}{\textbf{\begin{tabular}[c]{@{}c@{}}Data\\Setting\end{tabular}}}
& \multirow{3}{*}{\textbf{\begin{tabular}[c]{@{}c@{}}Search\\Scheme\end{tabular}}}
& \multirow{3}{*}{\textbf{\begin{tabular}[c]{@{}c@{}}Cryptographic\\Algorithm\end{tabular}}}
& \multirow{3}{*}{\textbf{\begin{tabular}[c]{@{}c@{}}Search\\Method\end{tabular}}}
& \multicolumn{3}{c}{\textbf{Privacy}}
& \textbf{Accuracy}
& \textbf{Efficiency}\\

\cmidrule(lr){6-8}
\cmidrule(lr){9-9}
\cmidrule(lr){10-10}

&
&
&
&
&
\textbf{\begin{tabular}[c]{@{}c@{}}Stored\\Data\end{tabular}}
& \textbf{\begin{tabular}[c]{@{}c@{}}Access\\Pattern\end{tabular}}
& \textbf{\begin{tabular}[c]{@{}c@{}}Search\\Pattern\end{tabular}}
& \textbf{\begin{tabular}[c]{@{}c@{}}Near-Plaintext\\Quality\end{tabular}}
& \textbf{\begin{tabular}[c]{@{}c@{}}Efficient Online\\Search\end{tabular}}\\
\midrule

CGKO06~\cite{curtmola2006searchable}
& Encrypted DB & Keyword
& SSE & Index-Based
& \cmark & \xmark & \xmark & -- & --\\

CLRZ18~\cite{chen2018differentially}
& Encrypted DB & Keyword
& SSE, DP & Index-Based
& \cmark & \cmark & \xmark & -- & --\\

SOPK21~\cite{Shang2021OSSE}
& Encrypted DB & Keyword
& SSE, DP & Index-Based
& \cmark & \cmark & \cmark & -- & --\\

\midrule

LZXL25~\cite{liu2025privacy}
& Encrypted DB & Embedding
& DCPE, DCE & Graph-Based
& \cmark & \xmark & \xmark & \cmark & \cmark\\

PANTHER~\cite{li2025panther}
& Encrypted DB & Embedding
& PIR, HE, SS, GC & $k$-means
& \cmark & \cmark & \cmark & \cmark & \xmark\\

Compass~\cite{zhu2025compass}
& Encrypted DB & Embedding
& ORAM & Graph-Based
& \cmark & \cmark & -- & \cmark & \xmark\\

Tiptoe~\cite{henzinger2023private}
& Server-Held DB & Embedding
& HE, PIR & $k$-means
& -- & \cmark & -- & \cmark & \xmark\\

Wally~\cite{asi2024scalable}
& Server-Held DB & Embedding
& HE, PIR, DP & $k$-means
& -- & \cmark & -- & \cmark & \xmark\\

PACMANN~\cite{zhou2024pacmann}
& Server-Held DB & Embedding
& PIR & Graph-Based
& -- & \cmark & -- & \cmark & \xmark\\
\hdashline
\textbf{Ours}
& Encrypted DB & Embedding
& LSHRR & LSH, Graph-Based
& \cmark & \cmark & \cmark & \cmark & \cmark\\

\bottomrule
\end{tabular}%
}

\vspace{1pt}
\begin{minipage}{\textwidth}
\footnotesize
\noindent
\textit{Abbreviations:}
DB: database; SSE: searchable symmetric encryption; DP: differential privacy;
DCPE: distance-comparison-preserving encryption; DCE: distance comparison
encryption; PIR: private information retrieval; HE: homomorphic encryption;
SS: secret sharing; GC: garbled circuit; ORAM: oblivious random access memory;
LSH: locality-sensitive hashing; and LSHRR: LSH randomized response.

\noindent
\textit{Notes:}
Encrypted DB denotes data encrypted before outsourcing, while Server-Held DB
denotes data owned by the server. \cmark indicates that the goal is supported;
\xmark indicates that it is not protected or achieved; and ``--'' denotes an
inapplicable or unreported guarantee. Near-plaintext quality means that the
reported retrieval quality approaches the corresponding plaintext ANN
baseline. Efficient online search means indexed, one-round retrieval without
HE-based similarity evaluation or ORAM-protected graph traversal.
\end{minipage}
\label{tab:related}
\end{table*}

The remainder of the paper is structured as follows. \autoref{sec:models} presents the system model and discusses the security goals. \autoref{sec:preliminaries} provides the relevant background, followed by \autoref{sec:scheme}, which describes \name{} in detail. \autoref{sec:analysis} presents the security analysis, followed by the performance evaluation in \autoref{sec:evaluation}. \autoref{sec:relatedwork} discusses related work, and \autoref{sec:conclusion} concludes. 
\section{System and Threat Models}\label{sec:models}

In this section, we first outline the architecture of \name{}. We then discuss the threat model, followed by the system and security goals. 
\subsection{System Overview}

\begin{figure}
    \centering
    \includegraphics[width=0.5\textwidth]{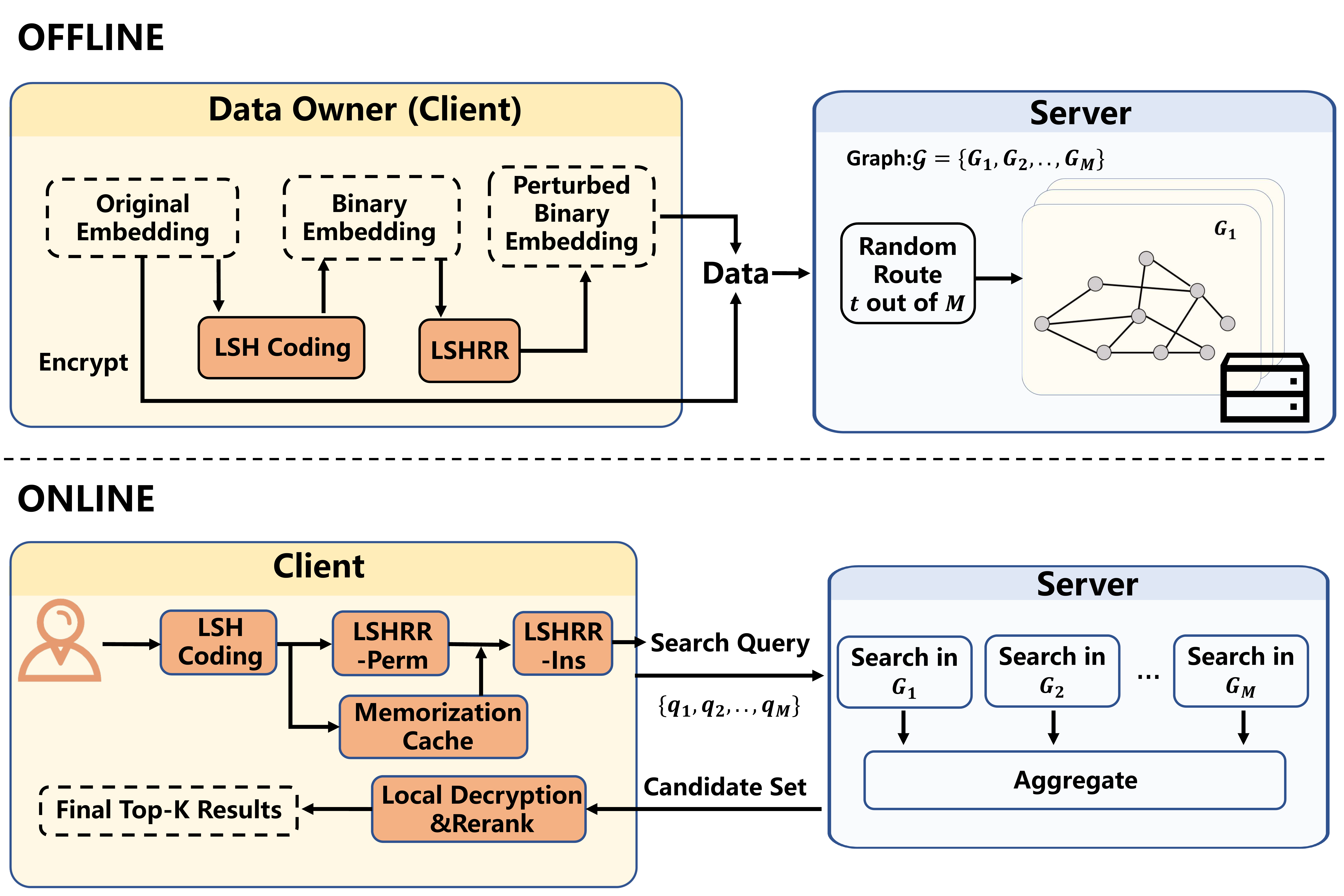}
    \caption{System architecture of \name.}
    \label{fig:overview}
\end{figure}
There are two entities in \name{}: a client $\mathcal{C}$ and a server $\mathcal{S}$. The client is the data owner and user; that is, it outsources the data to $S$ and issues search queries. The client's data is modeled as 
$\mathcal{D}=\{(\mathsf{id}_i,\bm{x}_i)\}_{i=1}^{N}$, where $id_i$ is the data item $i$ unique identifier, and $\bm{x}_i\in\mathcal{X}$ is its corresponding embedding vector. The client encrypts its data before sharing it with the server. 

\autoref{fig:overview} shows the system architecture of \name. It operates in two phases. In the offline phase, the client maps the original embeddings to binary codes using locality-sensitive hashing (LSH) and perturbs them using LSH randomized
response (LSHRR). The server maintains $M$ HNSW graph indexes (or {\em shards}). The client encrypts each data tuple and uploads it together with the perturbed codes to a subset of $M$. The server builds the HNSW indexes based on the perturbed codes.  In the online phase, the client hashes and perturbs the query embedding. The server performs a Hamming-distance search for nearest neighbors of the perturbed query embedding, across all the shards. It then aggregates the shard-level candidates and returns their encrypted data. The client decrypts the results, removes duplicates using the identifiers, and
reranks the candidates based on the decrypted embeddings.


\vspace{3pt}
\noindent \textbf{Scope.} 
\name{} supports semantic search over a database of vector embeddings, as opposed to directly over raw data (e.g., documents or images). We abstract away the data pre-processing pipeline that generates embeddings from raw data, and the post-processing pipeline for retrieving the raw data based on the search results. The embedding generation can be performed by off-the-shelf embedding models for different data modality ML models~\cite{izacard2021unsupervised}. \name{} can be deployed in a setting where the client keeps the original data locally and only uses the server to find similar data~\cite{zhu2025compass}. In this setting, the client uses the unique identifiers $id_i$ to retrieve the matching raw data from local storage. \name{} can also be integrated with other PIR schemes~\cite{simplepir} to retrieve raw data from the server without revealing which data items are being retrieved.

\subsection{Threat Model}
We consider a semi-honest server that follows the protocol but wants to learn information about the data embeddings and about the queries. It can infer any geometric and semantic information directly from the perturbed codes~\cite{Fernandes2021LSHXDP}. The server knows the LSH mapping. It can record
and analyze the complete view of the outsourced index and query execution. Specifically, 
during the offline phase, it can observe the encrypted data, the perturbed codes, shard assignment, and the HNSW graph structures. During the online phase, it has complete visibility into query execution, including graph traversal
traces, candidate sets, response sizes, and repeated ciphertexts. We assume that cryptographic primitives are secure, and that the server does not have access to client-side states. 

\subsection{Goals}
\name{} aims to achieve three properties: privacy, accuracy,
and efficiency. We further break down privacy into stored-data, access-pattern, and search-pattern privacy. We formally analyze privacy in~\autoref{sec:security-analysis}, and provide a detailed evaluation of the other two in~\autoref{sec:evaluation}.

\begin{itemize}[leftmargin=1.35em,itemsep=0.05em,topsep=0.1em]
\item \textbf{\textit{Stored-data privacy.}}
The server should not learn the contents of the original data tuple, consisting of the identifier or embedding, from the outsourced representation. Encrypted payloads should hide their plaintext contents, while the index incurs only well-defined leakage about the underlying embeddings.

\item 
\textbf{\textit{Access-pattern privacy.}}
The server's ability to distinguish two query embeddings based on their execution traces is bounded. In other words, the query execution leakage is bounded.  

\item 
\textbf{\textit{Search-pattern privacy.}}
Across multiple queries, the server should have limited ability to determine
whether different queries originate from the same query embedding value. In particular, it cannot identify repeated queries. Otherwise, it can group repeated submissions, observe their frequency and timing, and use auxiliary information to infer query values~\cite{blackstone2019revisiting,oya2021hiding}. 

\item 
\textbf{\textit{Accuracy.}}
The system enables high-recall similar search, despite LSH encoding and randomized perturbation. In other words, the search results are close to those of non-private ANN search. 

\item 
\textbf{\textit{Efficiency.}}
The search incurs small computation and communication overhead, resulting in end-to-end latency in practice. In other words, the search is performed without scanning the full datasets, without overhead of encrypted computation, and without multi-round communications. 
\end{itemize}
\vspace{3pt}
\noindent\textbf{Discussion.} 
Access-pattern privacy concerns leakage from a single execution, whereas search-pattern privacy concerns leakage across multiple executions. These two properties together provide a similar guarantee to query privacy in other private information retrieval (PIR) systems~\cite{simplepir,albab2022batched}.

\section{Preliminaries}\label{sec:preliminaries}

We use calligraphic uppercase letters for sets, spaces, and randomized
mechanisms. Bold lowercase and uppercase letters denote vectors and matrices,
respectively. Scalars are written in italic font. Let
$\mathcal X\subseteq\mathbb R^d$ denote the original embedding space and
$\mathcal D=\{(\mathsf{id}_i,\bm x_i)\}_{i=1}^{N}$ the original data, where
$\bm x_i\in\mathcal X$. A query embedding is denoted by
$\bm q\in\mathcal X$. We use $\|\bm x\|_2$ for the Euclidean norm and
$\langle\cdot,\cdot\rangle$ for the inner product.
$\theta_{\bm x,\bm x'}$ denotes the angle between two non-zero embeddings, and
$d_\theta(\bm x,\bm x')=\theta_{\bm x,\bm x'}/\pi$ denotes their normalized
angular distance.

$\mathcal H$ represents a LSH function family, and $h\in\mathcal H$ represents one function. Each graph index shard $s$ uses a fixed $\kappa$-bit mapping
$H_s:\mathcal X\rightarrow\mathcal V$, where
$\mathcal V=\{0,1\}^{\kappa}$ is the binary embedding space. We write
$\bm v_{i,s}=H_s(\bm x_i)$ and $\bm v_q^{(s)}=H_s(\bm q)$ for the binary
embeddings of a data item and query. A perturbed binary embedding is denoted
by $\widetilde{\bm v}$, and $d_H(\cdot,\cdot)$ denotes the Hamming distance.

\subsection{Locality-Sensitive Hashing under Angular Distance}

\noindent
LSH~\cite{indyk1998approximate} maps
high-dimensional vectors into compact binary codes such that vectors close to each other
are more likely to share hash values. In semantic search, vector similarity
is commonly measured by cosine similarity: 
$
\mathrm{sim}_{\cos}(\bm x,\bm x')
=
\frac{\langle\bm x,\bm x'\rangle}
{\|\bm x\|_2\|\bm x'\|_2}.
$

\smallskip\noindent
\textbf{Isotropic Hashing (IsoHash).}
IsoHash~\cite{kong2012isotropic} overcomes the problem of imbalanced or redundant hash bits when the vector distribution is anisotropic by learning an orthogonal transformation. In \name, IsoHash is trained separately for each graph shard. The learned
parameters are fixed and reused for both index construction and
query generation. We define IsoHash formally as 
\(\Pi_{\mathrm{IsoHash}}=(\mathrm{PCA},\mathrm{Rotate},\mathrm{Hash})\).

\begin{itemize}[leftmargin=*] 
\item
$(\mathbf W,\Lambda)\leftarrow
\Pi_{\mathrm{IsoHash}}.\mathrm{PCA}
(\mathbf X\in\mathbb R^{d\times n},\kappa)$:
Randomly select \(n\) samples to form the mean-centered training matrix
\(\mathbf X\), where \(\kappa\) denotes the hash length of each graph shard.
PCA returns the projection matrix
\(\mathbf W\in\mathbb R^{d\times\kappa}\), formed by the eigenvectors
corresponding to the largest \(\kappa\) eigenvalues, and the covariance matrix
of the projected dimensions:
$
\Lambda
=
\mathbf W^T\mathbf X\mathbf X^T\mathbf W
=
\mathrm{diag}(\lambda_1,\lambda_2,\ldots,\lambda_\kappa).
$

\noindent 
\item $\mathbf Q\leftarrow \Pi_{\mathrm{IsoHash}}.\mathrm{Rotate}( \Lambda\in\mathbb{R}^{\kappa\times \kappa})$: Given the covariance matrix of the projected dimensions $ \Lambda$, output an isotropic orthogonal rotation matrix $\mathbf Q\in\mathbb{R}^{\kappa\times \kappa}$ satisfying $\mathbf Q^T\mathbf Q=I$, such that the diagonal elements of $\mathbf Q^T \Lambda \mathbf Q$ are all equal, i.e., $ [\mathbf Q^T \Lambda \mathbf Q]_{11}=[\mathbf Q^T \Lambda \mathbf Q]_{22}=\cdots=[\mathbf Q^T \Lambda \mathbf Q]_{\kappa\kappa}=\bar{\lambda}. $ Since an orthogonal transformation preserves the trace of a matrix, the target mean variance is necessarily $ \bar{\lambda}=\frac{1}{\kappa}\sum_{i=1}^{\kappa}\lambda_i. $ This optimization is typically solved by Lift-and-Projection (LP) or Gradient Flow (GF) algorithms. 

\noindent 
\item $h\leftarrow \Pi_{\mathrm{IsoHash}}.\mathrm{Hash}(\bm x\in\mathbb{R}^d,\mathbf W,\mathbf Q)$: Given an input data point $\bm x$ together with matrices $\mathbf W$ and $\mathbf Q$, output the final $\kappa$-bit binary embedding $h$, defined by $ h=\mathrm{sgn}(\mathbf Q^T\mathbf W^T\bm x). $ 
\end{itemize}

\subsection{LSHRR-Based Embedding Perturbation}

Traditional $\epsilon$-differential privacy defines neighboring inputs through a binary adjacency relation. This is less suitable for semantic search, where the protected inputs are continuous embeddings and privacy loss should depend on their distance. We therefore use extended differential privacy (XDP), which is based on the distance induced by a fixed binary hash mapping. 
Each graph shard learns an independent hash mapping. Let
$H:\mathcal X\rightarrow\{0,1\}^{\kappa}$, where
$H(\bm x)=\bigl(h_1(\bm x),\ldots,h_{\kappa}(\bm x)\bigr)$.
For two embeddings $\bm x,\bm x'\in\mathcal X$, we define the induced Hamming
pseudometric as
$
d_H(\bm x,\bm x')
=
\sum_{i=1}^{\kappa}
\mathbf{1}[h_i(\bm x)\ne h_i(\bm x')].
$
LSHRR~\cite{Fernandes2021LSHXDP} applies randomized response independently to
each bit of $H(\bm x)$. Given a per-bit privacy parameter $\epsilon$, the
bit-flip probability is
$
p=1/(e^\epsilon+1).
$
The resulting mechanism $Q_H$ defines a distribution over
$\{0,1\}^{\kappa}$. For any output
$\bm y\in\{0,1\}^{\kappa}$,
$
\Pr[Q_H(\bm x)=\bm y]
=
\prod_{i=1}^{\kappa}
p^{|y_i-h_i(\bm x)|}
(1-p)^{1-|y_i-h_i(\bm x)|}.
$
It follows that  bitwise randomized response under a fixed binary hash mapping satisfies, for every
$\bm x,\bm x'\in\mathcal X$ and
$S\subseteq\{0,1\}^{\kappa}$,
$
\Pr[Q_H(\bm x)\in S]
\le
e^{\epsilon d_H(\bm x,\bm x')}
\Pr[Q_H(\bm x')\in S]
$.
In other words, $Q_H$ satisfies $(\epsilon d_H,0)$-XDP under the fixed hash-induced Hamming pseudometric. As a result, embeddings that differ in a few fixed-hash coordinates induce similar output distributions.

\subsection{Hierarchical Navigable Small World}

\noindent
The Hierarchical Navigable Small World (HNSW)~\cite{malkov2018efficient}
is a widely used vector index that achieves high search recall. It combines skip lists with navigable proximity graphs. Let $\mathcal{G}=\{G_0,G_1,\dots,G_L\}$ denote the $L+1$ graph layers. The graph construction is guided by four parameters: $M_{\mathrm{hnsw}}$ that limits the number of connections per node, with the bottom layer typically allowing about
$2M_{\mathrm{hnsw}}$ connections; $ef_{\mathrm{construction}}$ that controls the
candidate search during index construction;
$ef_{\mathrm{search}}$ that determines the recall--latency trade-off during
search; and $m_L$ that controls
the exponentially decreasing node distribution across layers, typically set to $1/\ln(M_{\mathrm{hnsw}})$.

The core operations are defined as follows. 
\begin{itemize}[leftmargin=*]
    \item
    $(\mathcal{G},ep_L)\leftarrow
    \Pi_{\mathrm{HNSW}}.\mathrm{InitGraph}()$:
    Initialize an empty hierarchical graph, its layers, and the entry
    point used by insertion and search.

    \item
    $\mathcal{G}'\leftarrow
    \Pi_{\mathrm{HNSW}}.\mathrm{Insert}
    (\mathcal{G},\bm v\in\mathcal V)$:
    Insert $\bm v$ and output the updated graph $\mathcal{G}'$. The algorithm
    samples its maximum level, greedily routes through the upper layers,
    searches local candidates using $ef_{\mathrm{construction}}$, selects
    neighbors bounded by $M_{\mathrm{hnsw}}$ (approximately
    $2M_{\mathrm{hnsw}}$ at the bottom layer), and creates bidirectional edges.

    \item
    $\mathcal{R}\leftarrow
    \Pi_{\mathrm{HNSW}}.\mathrm{Search}
    (\mathcal{G},\bm v_q\in\mathcal V,K)$:
    Given a query $\bm v_q$, the algorithm greedily starts
    from the top layer, searches the bottom layer with candidate budget
    $ef_{\mathrm{search}}$, and returns the approximate top-$K$ nearest-neighbor
    set $\mathcal{R}$.
\end{itemize}

\section{Design of \name}\label{sec:scheme}
In this section, we describe the design of \name{}, which achieves privacy, accuracy, and efficiency. \name{} adopts differential privacy, allowing it to bound the privacy leakage and avoid the overheads of other cryptographic approaches. However, searching over perturbed codes degrades accuracy because the distances between codes do not match those between the original embeddings. The system can improve accuracy by returning larger candidate sets to the client, but doing so increases communication and computational overhead. Below, we elaborate on the two challenges related to accuracy and discuss our approach. 

\smallskip\noindent
\textbf{Challenge 1: Candidate-set expansion under perturbation.}
High bit-flip probabilities provide strong privacy, but distort nearest-neighbor rankings. Our experiments confirm this effect: with a single perturbed graph, SIFT achieves only \(56.65\%\) recall using 1,500 candidates, while LAION achieves only \(69.6\%\) recall using 700 candidates. To improve recall, the server can return a larger candidate set for client-side reranking, but this incurs higher communication and computational overhead. Reducing the bit-flip probability could also improve recall, but the codes expose more semantic information, thereby reducing privacy. 

\noindent\textit{Our approach.}
\name employs multiple index graphs (or shards) to reduce competition among candidates within each shard. Randomized response may cause non-target points to overtake a true neighbor in Hamming space. Let \(N(d)\) be the number of points at distance \(d\), and let \(P_{\mathrm{overtake}}(d)\) be the probability that
such a point overtakes the true neighbor after perturbation. The expected
overtaking count in a single graph is
\(\mu_{\mathrm{single}}=\sum_{d=0}^{\kappa}
N(d)P_{\mathrm{overtake}}(d)\). When the codes are replicated uniformly in $t$ out of $M$ shards, 
\(\mu_{\mathrm{shard}}=(t/M)\mu_{\mathrm{single}}\). Each shard 
uses a small candidate set, and aggregation over all the shards 
gives the true neighbors multiple opportunities to be recovered.

\smallskip\noindent
\textbf{Challenge 2: Search instability under perturbation.}
Random bit flipping may severely displace the true neighbors of certain embeddings. A single graph may perform well for most queries, but require a prohibitively large candidate set for queries with severely displaced neighbors.

\noindent\textit{Our approach.}
\name builds multiple shards, each with an independent IsoHash mapping and
independently sampled randomized-response coins. A true neighbor that is poorly
positioned in one shard may remain close to the query in another. Let
\(P_{\mathrm{hit}}(b)\) be the recovery probability in one shard with
local candidate budget \(b\), under an approximate independence
assumption, the recovery probability from at least one of the \(t\) shards is
\(P_{\mathrm{multi}}(b)\approx
1-(1-P_{\mathrm{hit}}(b))^t\). 

\subsection{Multi-Graph Construction}

\noindent
Let
\(\mathcal D=\{(\mathsf{id}_i,\bm x_i)\}_{i=1}^{N}\subseteq\mathcal X\)
denote the original embeddings. The system maintains \(M\) logically separated HNSW graph shards,
denoted by
\(\mathbb G=\{\mathcal G^{(1)},\ldots,\mathcal G^{(M)}\}\). Each shard
\(\mathcal G^{(j)}\) stores only the perturbed index binary embeddings (or perturbed codes) and
the associated encrypted payloads.

~\autoref{fig:graphbuildprotocol} illustrate the multi-graph construction. 
For each embedding $i$, the client assigns it to a set of shards 
\(\mathcal S_i\leftarrow\mathsf{Route}(i,M,t)\). We instantiate
\(\mathsf{Route}\) as a data-independent pseudorandom function that
permutes \([M]\) and selects the first \(t\) shards. For every shard \(j\in\mathcal S_i\), the client computes the
shard-level binary embedding
\(\bm v_{i,j}\leftarrow
\Pi_{\mathrm{IsoHash}}.\mathsf{Hash}
(\bm x_i,\allowbreak\bm\mu_j,\allowbreak
\mathbf W_j,\allowbreak\mathbf Q_j)\), then applies LSHRR to obtain
\(\widetilde{\bm v}_{i,j}\leftarrow
\Pi_{\mathrm{LSHRR}}(\bm v_{i,j};p)\).
It encrypts the payload 
\(m_i=\mathsf{id}_i\Vert\bm x_i\) with the shard key to obtain 
\(\mathsf{ct}_{i,j}\leftarrow
\Pi_{\mathrm{Enc}}.\mathsf{Enc}_{sk_j}(m_i)\). The client also generates a shard-local label \(\ell_{i,j}\), used later for deletion. The resulting
tuple 
\((\widetilde{\bm v}_{i,j},\ell_{i,j},\mathsf{ct}_{i,j})\)
is inserted into \(\mathcal G^{(j)}\). 

In an ideal setting, the \(M\) graph shards are maintained by
\(M\) non-colluding servers. In this case, an attacker controlling a single server can observe a given embedding with probability \(t/M\). Its local view
can therefore exhibit a subsampling effect under the standard conditions for
differential-privacy amplification. \name{} assumes a single-server setting, in which the attacker has the complete view of all the embeddings. Therefore, we do not use \(t/M\) as a subsampling probability in the privacy analysis.

\subsection{Single-Server Multi-Graph}

\noindent
The server maintains all graph shards. \name employs several mechanisms to reduce cross-shard linkage.

\begin{itemize}[leftmargin=*]

\item \textit{Shard-specific projection functions.}
Each shard uses different IsoHash parameters. The same original embedding
\(\bm x_i\) is therefore mapped to different shard-level binary embeddings,
making bit coordinates across different shards not directly comparable.

\item \textit{Independent perturbation.}
Randomized response uses fresh coins for every shard. Even when the
same embedding appears in multiple shards, its binary embeddings are
independently perturbed, preventing direct equality comparison and reducing
the effectiveness of averaging attacks.

\item \textit{Randomized shard assignment.}
Each embedding is assigned to only \(t\) pseudorandomly selected shards instead of
all \(M\) shards. This limits the number of server-visible entries associated
with the same item and avoids fixed shard co-occurrence patterns.

\item \textit{Shard-level labels and encryption keys.}
Labels are generated independently for different shards, while payloads are
encrypted with shard-level keys and fresh randomness.
As a result, identifier or ciphertext equality cannot be
used to link embeddings across different shards. 

\end{itemize}
\smallskip
These mechanisms make cross-shard linkage more difficult for the attacker, but do not eliminate it.  Our formal analysis (~\autoref{sec:analysis}) accounts for all \(t\) perturbed codes originating from the same embedding. We also evaluate representative cross-shard attacks to measure whether the server can link embeddings across shards. 

\begin{figure}
	\centering
	\setlength{\fboxrule}{0.25pt}
	\setlength{\fboxsep}{4pt}
	\noindent\fbox{%
		\begin{minipage}{\dimexpr\columnwidth-2\fboxsep-2\fboxrule\relax}
			\vspace{1.5mm}
			
			\noindent \textbf{Algorithmic Primitives:}
			Isotropic hashing $\Pi_{\mathrm{IsoHash}}$, locality-sensitive hashing
			randomized response $\Pi_{\mathrm{LSHRR}}$, symmetric encryption
			$\Pi_{\mathrm{Enc}}$, HNSW indexing $\Pi_{\mathrm{HNSW}}$.\\
			\textbf{Input:}
			Plaintext embedding database
            {$\mathcal D=\{(\mathsf{id}_i,\bm x_i)\}_{i=1}^{N}$ with
            $\bm x_i\in\mathcal X$, }
			shard-level binary embedding length $\kappa$, total number of graph
			shards $M$, routing multiplicity $t$, randomized-response flip
			probability $p$, and HNSW parameters.\\
			\textbf{Output:}
			Public configuration $\mathsf{pp}$, client-side state
			$\mathsf{st}_{\mathcal C}$, and server-side multi-graph
			index $\mathsf{st}_{\mathcal S}$.
			
			\vspace{1.5mm}
			\hrule
			\vspace{1.5mm}
			
			\noindent \textbf{[Phase 1: Parameter Setup]}
			\begin{enumerate}[label=\arabic*:, leftmargin=0.6cm, itemsep=0.7mm, topsep=1mm]
				\item The client samples representative training data from
				$\mathcal D$ and trains shard-specific IsoHash parameters
				$\mathsf{st}_{\mathrm{hash}}=\{(\bm\mu_j,\mathbf W_j,\mathbf Q_j)\}_{j=1}^{M}$.
				The public configuration is set as
				$\mathsf{pp}=(\kappa,M,t,p,\mathsf{params}_{\mathrm{HNSW}})$.
			\end{enumerate}
			
			\vspace{1mm}
			\noindent \textbf{[Phase 2: Multi-Graph Initialization]}
			\begin{enumerate}[label=\arabic*:, leftmargin=0.6cm, itemsep=0.7mm, topsep=1mm, start=2]
				\item The client prepares $M$ independent symmetric encryption
				keys and forms the shard key set $\mathcal K=\{sk_1,\dots,sk_M\}$.
				
				\item The cloud server initializes $M$ logically separated HNSW graph
				shards. The outsourced multi-graph index is denoted by
				$\mathbb G=\{\mathcal G^{(1)},\ldots,\mathcal G^{(M)}\}$, where each
				shard $\mathcal G^{(j)}=\{G^{(j)}_0,\ldots,G^{(j)}_{L_j}\}$ is an
				independent multi-layer HNSW index.
			\end{enumerate}
			
			\vspace{1mm}
\noindent \textbf{[Phase 3: Perturbation and Insertion]}
\begin{enumerate}[label=\arabic*:, leftmargin=0.6cm, itemsep=0.7mm, topsep=1mm, start=4]
    \item For each data point $(id_i,\bm x_i)$, the client obtains its selected
    shard set by calling $S_i \leftarrow Route(id_i,M,t)$.
    
    \item For each selected shard $j\in\mathcal S_i$, the client
    computes the shard-level binary embedding
    $\bm v_{i,j}\leftarrow\Pi_{\mathrm{IsoHash}}.\mathrm{Hash}
    (\bm x_i,\bm\mu_j,\mathbf W_j,\mathbf Q_j)$, and perturbs it as
    $\widetilde{\bm v}_{i,j}\leftarrow
    \Pi_{\mathrm{LSHRR}}(\bm v_{i,j};p)$.
    
    \item The client generates a shard-local label $\ell_{i,j}$ and 
    {forms the encrypted payload
    $m_i \leftarrow \mathsf{id}_i \Vert \bm x_i$, where $\mathsf{id}_i$
    is used for duplicate removal and result recovery, and $\bm x_i$ is
    used for client-side exact reranking. The client encrypts $m_i$
    under $sk_j$ with fresh encryption randomness, yielding
    $\mathsf{ct}_{i,j}\leftarrow
    \Pi_{\mathrm{Enc}}.\mathrm{Enc}_{sk_j}(m_i)$.}
    
    \item The tuple $(\widetilde{\bm v}_{i,j},\ell_{i,j},\mathsf{ct}_{i,j})$ is
    inserted into the HNSW graph shard $\mathcal G^{(j)}$
    by running
    $\mathcal G^{(j)}\leftarrow
    \Pi_{\mathrm{HNSW}}.\mathrm{Insert}
    (\mathcal G^{(j)},\widetilde{\bm v}_{i,j},\ell_{i,j},\mathsf{ct}_{i,j})$.
\end{enumerate}
			
			\vspace{1mm}
			\noindent \textbf{[Phase 4: State Finalization]}
			\begin{enumerate}[label=\arabic*:, leftmargin=0.6cm, itemsep=0.7mm, topsep=1mm, start=8]
				\item The protocol outputs
				$\mathsf{st}_{\mathcal C}=(\mathsf{st}_{\mathrm{hash}},\mathcal K)$
				and $\mathsf{st}_{\mathcal S}=\mathbb G$.
			\end{enumerate}
			\vspace{1mm}
		\end{minipage}%
	}
	\caption{Multi-graph index construction protocol 
	($\Pi_{\mathsf{GraphBuild}}$).}
	\label{fig:graphbuildprotocol}
\end{figure}

\begin{figure}
    \centering
    \includegraphics[width=0.45\textwidth]{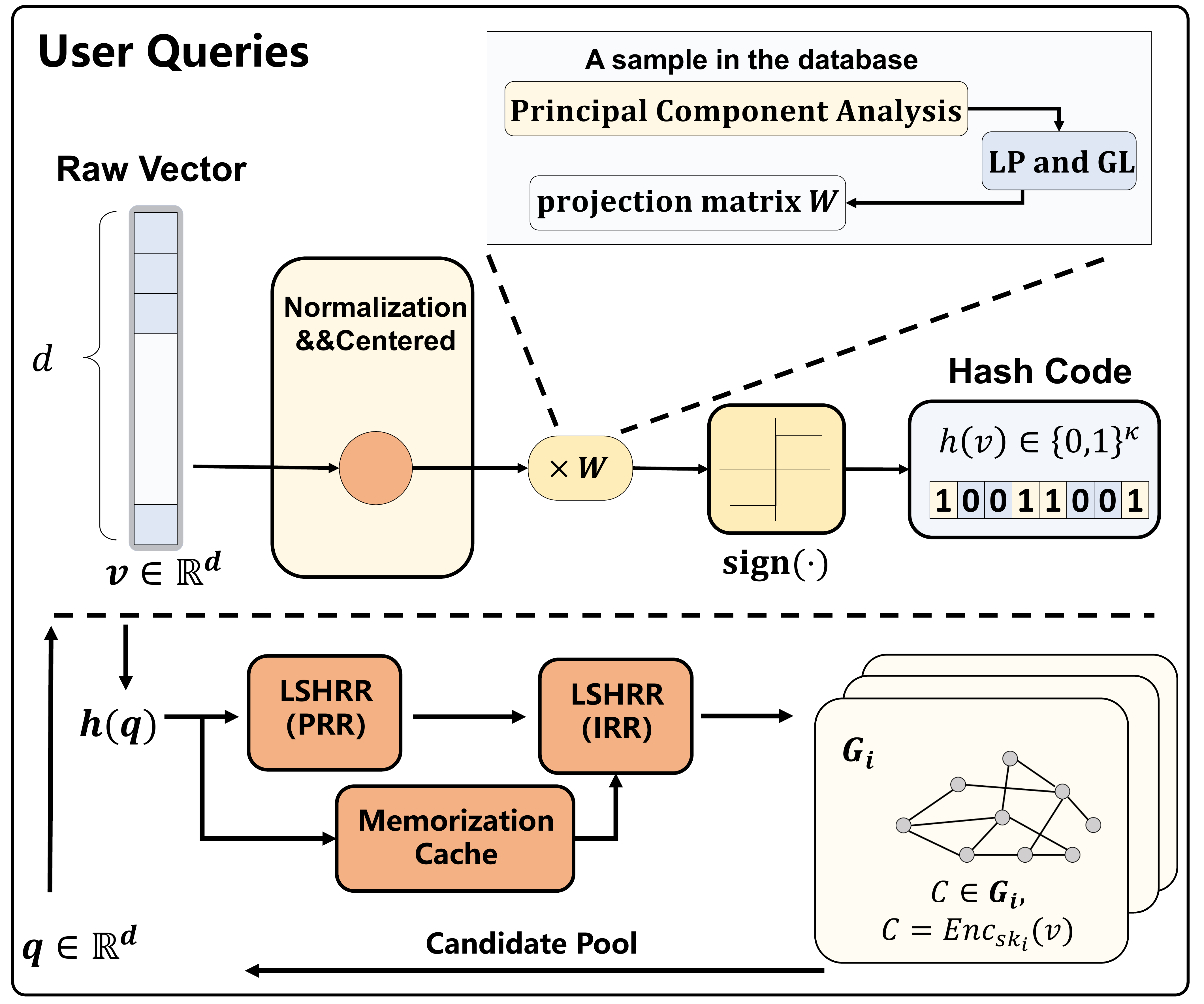}
    \caption{Systematic illustration of the query flow.}
    \label{fig:query}
\end{figure}

\subsection{Processing User Queries}

\noindent
Repeated queries may reveal the underlying value if each query is independently perturbed. By collecting many versions of the same query, the attacker can average out the
randomness and infer the original binary embedding. To mitigate
this averaging attack, \name adopts a two-phase LSHRR mechanism, inspired by the memoization technique proposed by RAPPOR~\cite{erlingsson2014rappor}. \autoref{fig:query} illustrates the query workflow.

\smallskip\noindent\textbf{Permanent LSHRR. }
For a query $\bm q$, let $\bm v_{\bm q}$ denote its unperturbed binary
embedding, and let $p_{\mathrm{PRR}}$ denote the permanent flip probability. In
this phase, the client first checks its local cache. If $\bm q$ is
issued for the first time, the client samples
\[
\widetilde{\bm v}_{\bm q}^{(\mathrm{PRR})}
\leftarrow
\Pi_{\mathrm{LSHRR}}(\bm v_{\bm q},p_{\mathrm{PRR}}),
\]
and writes 
$\widetilde{\bm v}_{\bm q}^{(\mathrm{PRR})}$ to the cache. If $\bm q$ is in the cache, $\widetilde{\bm v}_{\bm q}^{(\mathrm{PRR})}$ is reused. As a result,
repeated submissions do not have independent perturbations of the
original query binary embedding.

\smallskip\noindent\textbf{Instantaneous LSHRR. }Before sending the query, the client applies an instantaneous perturbation
\[
\widetilde{\bm v}_{\bm q}^{(\mathrm{IRR})}
\leftarrow
\Pi_{\mathrm{LSHRR}}
\bigl(\widetilde{\bm v}_{\bm q}^{(\mathrm{PRR})},p_{\mathrm{IRR}}\bigr),
\]
where $p_{\mathrm{IRR}}$ is the instantaneous flip probability.  The server uses
$\widetilde{\bm v}_{\bm q}^{(\mathrm{IRR})}$ as input to the search. Since $\widetilde{\bm v}_{\bm q}^{(\mathrm{IRR})}$ varies across repeated queries, the server cannot link them. 

\subsection{\name Protocol}

We now combine the graph construction and query protocol above into the complete protocol.
\name{} consists of four key operations: initialization, secure query generation, search, and
decryption and reranking. It also supports update operations, i.e., insertion and deletion.

\begin{enumerate}[leftmargin=*]

\item \emph{Initialization.}
Given $\mathcal D=\{(\mathsf{id}_i,\bm x_i)\}_{i=1}^{N}$, the client trains shard-specific hashing parameters, generates shard keys, and initializes
$\mathbb G=\{\mathcal G^{(1)},\ldots,\mathcal G^{(M)}\}$. Each item is assigned to a set of shards 
$\mathcal S_i\leftarrow
\mathsf{Route}(\mathsf{id}_i,M,t)$.
For every $j\in\mathcal S_i$, the client computes and perturbs the shard-level binary embedding, encrypts $m_i=\mathsf{id}_i\Vert\bm x_i$ under $sk_j$, and sends the resulting values to the server. The server inserts these values into the corresponding HNSW shards. The resulting states are $\mathsf{st}_{\mathcal C}=(\mathsf{st}_{\mathrm{hash}},\mathcal K)$ and
$\mathsf{st}_{\mathcal S}=\mathbb G$.

\item \emph{Secure query generation.}
Given a query embedding $\bm q$, the client derives a query binary embedding for each shard. It applies permanent perturbation to prevent averaging across repeated submissions, followed by fresh instantaneous perturbation for each query. The resulting perturbed query embeddings are sent to the server.

\item \emph{Search.}
The server searches the HNSW shards using Hamming distance between the perturbed query and index nodes. It aggregates the shard-level
candidate sets and returns their encrypted payloads. 

\item \emph{Decryption and reranking.}
The client decrypts the returned payloads, removes duplicates using the identifiers, and
reranks the remaining candidates using the original query and data 
embeddings. It then outputs the final top-$K$ results.

\item \emph{Insertion.}
To insert a new tuple $(\mathsf{id}_i,\bm x_i)$, the client reuses the existing hashing parameters and shard keys. It selects $\mathcal S_i$, generates independently perturbed binary embeddings and encrypted payloads for the selected shards, and computes the shard-local label $\ell_{i,j}$. The server inserts the entries into the corresponding HNSW shards, while the client stores
$\{(j,\ell_{i,j})\}_{j\in\mathcal S_i}$ for future deletion. Insertion does not require rebuilding the index.

\item \emph{Deletion.}
The client retrieves the shard-local labels associated with
$\mathsf{id}_i$ and sends deletion requests to the server. The
server marks the corresponding index nodes as inactive and excludes them from the search. The server performs periodic index rebuilding to remove inactive nodes.
\end{enumerate}



\section{Theoretical Analysis}\label{sec:analysis}
\subsection{Correctness Analysis}

We characterize how perturbation affects shard-level ranks and candidate
recovery.

\begin{proposition}[Approximate Candidate Recovery]
\label{prop:candidate-correctness}
\normalfont
Consider a target item in shard \(s\) at clean Hamming distance \(d_s^*\)
from the query, and let \(N_s(d)\) be the number of background items at
distance \(d\). The effective query flip probability is
\(p_q=p_{\mathrm{PRR}}(1-p_{\mathrm{IRR}})
+p_{\mathrm{IRR}}(1-p_{\mathrm{PRR}})\), and its combination with the
storage-side perturbation is
\(p_{\oplus}=p_D(1-p_q)+p_q(1-p_D)\).
A background item at distance \(d\) overtakes the target with approximate
probability
\[
P_{\mathrm{overtake},s}(d)
\approx
\Phi\!\left(
\frac{(d_s^*-d)(1-2p_{\oplus})}
{\sqrt{2\kappa p_{\oplus}(1-p_{\oplus})}}
\right).
\]
Under independent overtaking events, their count is Poisson binomial with mean
\(\mu_s=\sum_d N_s(d)P_{\mathrm{overtake},s}(d)\) and variance
\(\sigma_s^2=\sum_d N_s(d)P_{\mathrm{overtake},s}(d)
(1-P_{\mathrm{overtake},s}(d))\).
When sufficiently many items contribute, the target rank satisfies
\(Rank_s\approx\mathcal N(1+\mu_s,\sigma_s^2)\). For shard-level candidate
budget \(P\), its inclusion probability is
\[
P_{\mathrm{hit},s}(P)
\approx
\Phi\!\left(
\frac{P+\frac12-(1+\mu_s)}{\sigma_s}
\right).
\]
Across conditionally independent selected shards, the aggregate recovery
probability is approximately
\(1-\prod_{s\in\mathcal S_i}(1-P_{\mathrm{hit},s}(P))\).
If the aggregate candidate set contains the true Top-\(K\) items, client-side
exact reranking returns the plaintext Top-\(K\) result.
\end{proposition}

\noindent\textit{Proof Sketch.}
Among the \(d\) coordinates where the clean embeddings differ, the perturbed
bits remain different with probability \(1-p_{\oplus}\), so their mismatch
count follows
\(X_d\sim\operatorname{Binomial}(d,1-p_{\oplus})\).
Among the remaining \(\kappa-d\) coordinates, perturbation creates a mismatch
with probability \(p_{\oplus}\), giving
\(Y_d\sim\operatorname{Binomial}(\kappa-d,p_{\oplus})\).
Independence across bits gives
\(\widetilde D_d=X_d+Y_d\), with mean
\(\kappa p_{\oplus}+d(1-2p_{\oplus})\) and variance
\(\kappa p_{\oplus}(1-p_{\oplus})\).

Under the standard independence approximation, the difference between the
perturbed distances of a background item and the target has mean
\((d-d_s^*)(1-2p_{\oplus})\) and variance
\(2\kappa p_{\oplus}(1-p_{\oplus})\). Its Gaussian approximation gives
\(P_{\mathrm{overtake},s}(d)\).

For each background item \(i\) at distance \(d\), let
\(I_{i,d}=\mathbf 1\{\widetilde D_{i,d}<\widetilde D_{d_s^*}\}\).
Then
\(Rank_s=1+\sum_d\sum_{i=1}^{N_s(d)}I_{i,d}\).
The indicators have distance-dependent Bernoulli parameters, so their sum is
Poisson binomial with the stated mean and variance. Its normal approximation,
with the \(1/2\) continuity correction, gives
\(P_{\mathrm{hit},s}(P)\). Conditional independence across shards gives the
aggregate recovery probability. Finally, exact distance computation and
reranking return the plaintext Top-\(K\) result whenever all true Top-\(K\)
items are present.
\hfill\(\square\)

This approximation captures the effects of perturbation, candidate budget,
and multi-shard recovery. The detailed derivation of the shard-level rank
distribution and its connection to Recall and MRR are provided in
Appendix~\ref{app:rank-derivation}.

\subsection{Security and Privacy of the Server View}
\label{sec:security-analysis}

The complete server view contains all graph shards, perturbed index binary embeddings, shard-local HNSW structures, perturbed query reports, and search
traces. We protect stored embedding values, submitted query values, and the reuse relation among repeated queries, corresponding to stored-data,
access-pattern, and search-pattern privacy. The term \emph{complete server view} specifies the adversarial observation rather than a new privacy definition.

Our security statement is hybrid. Differential privacy protects the randomized index and query reports together with information derived from
them, while randomized authenticated encryption protects payload contents. These guarantees do not depend on the success of a concrete inference method.

\subsubsection{Scope and Fixed ISO-LSH Mappings}

For stored data, we use substitution adjacency: two databases contain the same identifiers and differ only in the embedding associated with one identifier. Thus, the guarantee protects the embedding value rather than the presence of its identifier.
Each shard $G_s$ uses a pretrained mapping
$H_s:\mathcal X\rightarrow\{0,1\}^{\kappa_s}$ fixed before the protected database is instantiated. The mappings may be correlated; the proofs require fresh randomized-response coins for each record, shard, and coordinate. If ISO-LSH is trained on the protected database, its training requires a separate privacy analysis. We therefore assume that the training data are public, disjoint from the protected database, or excluded from the
neighboring relation.

A mechanism $\mathcal M$ satisfies $(\xi,\delta)$-XDP if, for every pair
$x,x'$ and measurable event $E$,
\[
\Pr[\mathcal M(x)\in E]
\leq e^{\xi(x,x')}\Pr[\mathcal M(x')\in E]+\delta(x,x').
\]
Our guarantee uses the Hamming pseudometric induced by the fixed ISO-LSH mappings. It differs from angular-XDP averaged over freshly sampled random-LSH functions~\cite{Fernandes2021LSHXDP}; the collision distribution from that sampling is not applied to the pretrained mappings.

\subsubsection{Stored Multi-Graph Index}

Let $p_D\in(0,1/2)$ be the storage-side bit-flip probability; define
$\eta_D=\ln((1-p_D)/p_D)$, and let $d_{H_s}(\bm x,\bm x')=d_H(H_s(\bm x),H_s(\bm x'))$.

\begin{proposition}[Single-Shard XDP]
\label{prop:single-shard-privacy}
\normalfont
Conditioned on a record being assigned to $G_s$, bitwise randomized response
satisfies $\eta_D$-XDP with respect to $d_{H_s}$. A fixed pair
$q=(\bm x_0,\bm x_1)$ therefore has shard-level parameter
$\varepsilon_{s,q}=\eta_Dd_s(q)$, where
$d_s(q)=d_H(H_s(\bm x_0),H_s(\bm x_1))$.
{
Thus, if two embeddings differ in $d$ fixed-hash coordinates, the shard-local observation differs by at most a factor of $e^{\eta_D d}$. A smaller hash distance results in stronger privacy.}
\end{proposition}

\begin{proof}[Skeched Proof]
Equal clean bits induce identical output distributions. For a differing bit,
the likelihood ratio is at most
$(1-p_D)/p_D=e^{\eta_D}$. Independence across coordinates gives
\[
\frac{\Pr[\bm Y_s(\bm x)=\bm y]}
     {\Pr[\bm Y_s(\bm x')=\bm y]}
\leq e^{\eta_Dd_{H_s}(\bm x,\bm x')}.
\]
Summing over outputs in $E$ proves the claim; exchanging $\bm x$ and
$\bm x'$ proves the reverse direction.
\end{proof}

Each record is assigned to exactly $t$ of the $M$ shards independently of its embedding. Let $\mathfrak R$ contain the $t$-shard sets that may be selected;
for identifier-fixed assignment, it contains only the realized set. Define
$
D_D(\bm x,\bm x')
=\max_{\mathcal R\in\mathfrak R}
\sum_{s\in\mathcal R}d_{H_s}(\bm x,\bm x').
$

\begin{theorem}[Complete Stored-Index Privacy]
\label{thm:complete-view-storage}

For every measurable event $E$,
\[
\Pr[\mathsf V_D(\bm x)\in E]
\leq
e^{\eta_DD_D(\bm x,\bm x')}
\Pr[\mathsf V_D(\bm x')\in E],
\]
with the symmetric inequality obtained by exchanging $\bm x$ and $\bm x'$.
Hence, the stored multi-graph index is $\eta_D$-XDP with respect to $D_D$.
If every shard has $\kappa$ bits, it also satisfies
$(t\kappa\eta_D,0)$-DP under unrestricted substitution adjacency.
{
In other words, for two stored embeddings $\bm x$ and $\bm x'$, the complete multi-graph observation differs by at most a
factor of $e^{\eta_DD_D(\bm x,\bm x')}$, where $D_D$ aggregates the
differing hash coordinates across the selected shards.
}
\end{theorem}

\begin{proof}[Skeched Proof]
Condition on a selected set $\mathcal R$. Only the $t$ corresponding shard-local entries depend on the substituted embedding.
Proposition~\ref{prop:single-shard-privacy} and composition give
$\eta_D\sum_{s\in\mathcal R}d_{H_s}(\bm x,\bm x')$.
For randomized assignment, its distribution is identical in both worlds; averaging the conditional inequalities and bounding each sum by $D_D$ proves the result. The fixed-assignment case follows directly.

Each HNSW shard is constructed from perturbed codes, public parameters, and world-independent randomness. Its topology and subsequent inference outputs are therefore post-processing and incur no additional privacy cost.
\end{proof}

The theorem conservatively grants the server the selected shard set and perfect association of the changed record's shard-local entries. Subsampling amplification is not applied to the complete view: although one prespecified shard contains the record with probability $t/M$, the complete server observes exactly $t$ releases.

\smallskip\noindent\textbf{Pair-Specific XDP Accounting.}
For pair $q=(\bm x_0,\bm x_1)$ and selected set $\mathcal R$, let
$u_{\mathcal R}(q)=\sum_{s\in\mathcal R}d_s(q)$. Its route-conditioned pure-XDP value is $\xi_{\mathcal R}(q)=\eta_Du_{\mathcal R}(q)$. For randomized assignment, define
$w_q(u)=\Pr_{\mathcal R}[u_{\mathcal R}(q)=u]$; for fixed assignment,
$w_q$ is a point mass.

Conditioned on $u$ differing coordinates,
$Z\sim\operatorname{Binomial}(u,1-p_D)$ and
$L_u=(2Z-u)\eta_D$. The corresponding privacy profile is
\[
\delta_u^{\mathrm{RR}}(\varepsilon)
=
\sum_{z=0}^{u}
\binom{u}{z}(1-p_D)^zp_D^{u-z}
\left[1-e^{\varepsilon-(2z-u)\eta_D}\right]_+ .
\]

\begin{corollary}[Pair-Specific XDP Profile and Upper Tail]
\label{cor:stored-pld-tail}
\normalfont
For target $\delta_0$, the evaluated pair has approximate-XDP value
$
\xi_{D,q}^{\mathrm{PLD}}(\delta_0)
=\inf\{\varepsilon:
\sum_u w_q(u)\delta_u^{\mathrm{RR}}(\varepsilon)\leq\delta_0\}.
$
Moreover,
\[
\Pr[L_{D,q}^{\mathrm{impl}}>\varepsilon]
\leq
\min\!\left\{1,\,
\inf_{0\leq\alpha<\varepsilon}
\frac{\delta_{D,q}^{\mathrm{PLD}}(\alpha)}
     {1-e^{\alpha-\varepsilon}}
\right\},
\]
where
$\delta_{D,q}^{\mathrm{PLD}}(\alpha)
=\sum_u w_q(u)\delta_u^{\mathrm{RR}}(\alpha)$.

\end{corollary}

\begin{proof}[Skeched Proof]
Consider an augmented release that also reveals the challenge record, its selected shards, and their ground-truth correspondence. Conditioned on $u$, the binomial expression is its exact hockey-stick profile. The assignment has the same distribution in both worlds, so averaging over $w_q$ yields the augmented profile. Removing the additional correspondence information is post-processing and cannot increase hockey-stick divergence.

For $0\leq\alpha<\varepsilon$, $L>\varepsilon$ implies
$[1-e^{\alpha-L}]_+\geq1-e^{\alpha-\varepsilon}$. Taking expectations and optimizing over $\alpha$ proves the tail bound; randomized-response symmetry gives the reverse direction.
\end{proof}

We evaluate $\xi_{D,q}^{\mathrm{PLD}}(\delta_0)$ using the actual shard-specific distances of real neighboring pairs. The maximum over a finite challenge set applies only to those pairs; \autoref{thm:complete-view-storage} remains the unrestricted XDP guarantee.

\smallskip\noindent\textbf{Cross-Shard Linkage Evaluation.}
The formal analysis assumes perfect association of all $t$ shard-local entries belonging to the changed record. Our experiments test how much of this XDP evidence can be organized using MDS, Topology, GSGM, and their cycle-consistent joint combination.

For method $\mathcal A$, let $C_{\mathcal A}(\mathsf V)$ be its predicted cross-shard cluster and let $v_s^\star$ denote the target entry in $G_s$.
Ground truth is used only after inference. The distance of correctly associated target entries is
\[
U_{\mathcal A,q}(\mathsf V)
=
\sum_{(s,v)\in C_{\mathcal A}(\mathsf V)}
\mathbf 1[v=v_s^\star]\,d_s(q),
\]
giving recovered pure-XDP evidence
$\xi_{\mathcal A,q}^{\mathrm{rec}}
=\eta_DU_{\mathcal A,q}$. Recovering all $t$ target entries reaches the
route-conditioned complete-view value, whereas retaining one anchor gives
only its single-shard value. Failed and false matches contribute zero direct
target distance.

For each challenge pair, we construct two neighboring indexes that differ only
in one stored embedding. Each method receives the server-visible view and fixed
auxiliary seeds, and outputs a predicted cross-shard cluster and a
distinguishing score. Because the inferred cluster is data-dependent, we
evaluate matching accuracy, AUC, distinguishing advantage, and empirical
hockey-stick profiles. Inverting the profile yields an empirical pair-specific
XDP value for each evaluated method without replacing the formal guarantee.
Complete procedures, parameters, and results are provided in
Appendix~\ref{app:empirical-leakage-evaluation}.

\smallskip\noindent\textbf{Payload Confidentiality.}
Identifiers and original embeddings are stored in equal-length payloads using randomized authenticated encryption. Under IND-CPA security, fresh nonces and independent shard keys prevent ciphertext equality from directly associating entries of the same record across shards, assuming no common identifier or equivalent metadata is exposed.

\subsubsection{Access-Pattern Privacy}

For query-value privacy, the released index, client identity, submission schedule, PRR-reuse pattern, and search configuration are fixed, while a single logical query value is substituted. The server observes the perturbed query reports and complete HNSW access transcript. Our goal is to bound the query information revealed by these visible accesses, rather than to provide ORAM-style access-pattern obliviousness.

Let $a=p_{\mathrm{PRR}}$ and $b=p_{\mathrm{IRR}}$, where
$a,b\in(0,1/2)$. A permanent report is generated once for each logical-query coordinate and reused, whereas every submission applies fresh IRR coins.
Randomness is independent across coordinates and shards. For one clean bit $v$ submitted $R$ times, let $n_v(\bm y)$ be the number of reports in $\bm y=(y_1,\ldots,y_R)$ equal to $v$. The sequence probability and maximum symmetric log-likelihood ratio are
\[
P_v^{(R)}(\bm y)
=(1-a)(1-b)^{n_v}b^{R-n_v}
  +ab^{n_v}(1-b)^{R-n_v}\]
\[\eta_R
=\max_{\bm y}
\left|\ln\frac{P_v^{(R)}(\bm y)}
                 {P_{1-v}^{(R)}(\bm y)}\right|
=\ln\frac{(1-a)(1-b)^R+ab^R}
           {a(1-b)^R+(1-a)b^R}.
\]
The maximum occurs when all $R$ reports agree. Hence, $\eta_R$ jointly accounts for the memoized PRR state and all IRR reports, rather than applying basic composition across submissions.

For a fixed contacted shard set $\mathcal S_Q$, define
$D_Q(\bm q,\bm q')=
\sum_{s\in\mathcal S_Q}d_H$ $(H_s(\bm q),H_s(\bm q'))$.
If the set is sampled independently of the clean query, the pure-XDP distance is the maximum of this sum over all possible contacted sets. Routing based on
an unperturbed query is outside the theorem's scope.

\begin{theorem}[Complete Access-Transcript Privacy]
\label{thm:query-transcript}
\normalfont
For every measurable event $E$,
\[
\Pr[\mathsf V_{\mathrm{AP}}(\bm q)\in E]
\leq
e^{\eta_RD_Q(\bm q,\bm q')}
\Pr[\mathsf V_{\mathrm{AP}}(\bm q')\in E],
\]
with the reverse inequality obtained by exchanging $\bm q$ and $\bm q'$.
Thus, the complete access transcript is $\eta_R$-XDP with respect to $D_Q$.
{
Therefore, for any set of node-access traces, the probability of it being the trace of 
$\bm q$ is at most $e^{\eta_RD_Q(\bm q,\bm q')}$ times the probability of it being the trace of 
$\bm q'$. Observing the HNSW traversal thus reveals no more about the query
than the perturbed codes already do.}
\end{theorem}

\begin{proof}[Skeched Proof]
Equal clean bits induce identical sequence distributions, while each differing coordinate contributes at most $\eta_R$. Composition over the
differing coordinates and contacted shards gives
$\eta_RD_Q(\bm q,\bm q')$ for the complete query reports. Given the fixed index, HNSW receives only these reports. Visited nodes, traversal order,
candidate evolution, returned handles, response sizes, and cross-shard correspondences inferred from the index and traces are therefore post-processing and incur no additional privacy cost. Since the reports are
included in the server view and recoverable by projection, the complete access transcript and report mechanism have the same privacy profile.
\end{proof}

Accordingly, no separate access-pattern experiment is required: its XDP budget is computed directly as
$\xi_{\mathrm{AP}}(\bm q,\bm q')=\eta_RD_Q(\bm q,\bm q')$.
The guarantee requires every trace field to depend only on the perturbed reports, the fixed index, public parameters, and query-independent randomness. Clean-query routing, clean-distance calculations, stable query
identifiers, or external side channels must be removed or analyzed separately.

\subsubsection{Search-Pattern Privacy}

Search-pattern adjacency compares workloads with the same number and timing of submissions, contacted shards, and search configuration, but different
query-reuse relations. For $R\geq2$, all submissions in
$W_0=(\bm q,\ldots,\bm q)$ share one permanent report. In $W_1=(\bm q,\ldots,\bm q,\bm q')$, the final submission belongs to a new logical query and uses an independently sampled permanent report.

For one coordinate, let $r_v(z)$ be the PRR probability of permanent bit $z$ given clean bit $v$, and let $i_z(y)$ be the IRR probability of report $y$
given $z$. The reuse and split distributions are
\[
\begin{aligned}
P_{\mathrm{reuse}}^v(\bm y)
&=\sum_z r_v(z)\prod_{j=1}^{R}i_z(y_j),\\
P_{\mathrm{split}}^{v,v'}(\bm y)
&=\left(\sum_z r_v(z)\prod_{j=1}^{R-1}i_z(y_j)\right)
  \left(\sum_{z'}r_{v'}(z')i_{z'}(y_R)\right).
\end{aligned}
\]
Let $\gamma_{=}$ be their maximum symmetric log-likelihood ratio over $v=v'$ and all $\bm y$, and define $\gamma_{\ne}$ analogously for $v\neq v'$. Both values are computed exactly from $a$, $b$, and $R$.

\begin{theorem}[Complete Search-Pattern XDP]
\label{thm:search-pattern-transcript}
\normalfont
Let $K_Q=\sum_{s\in\mathcal S_Q}\kappa_s$ and
$u=D_Q(\bm q,\bm q')$. Define
\[
\xi_{\mathrm{SP}}(W_0,W_1)
=
(K_Q-u)\gamma_{=}+u\gamma_{\ne}.
\]
For every measurable event $E$,
\[
\Pr[\mathsf V_{\mathrm{SP}}(W_0)\in E]
\leq
e^{\xi_{\mathrm{SP}}(W_0,W_1)}
\Pr[\mathsf V_{\mathrm{SP}}(W_1)\in E],
\]
with the reverse inequality obtained by exchanging $W_0$ and $W_1$. Therefore, the complete transcript satisfies $(\xi_{\mathrm{SP}},0)$-XDP over the declared workload adjacency. A uniform ordinary-DP bound is
$\bigl(K_Q\max\{\gamma_{=},\gamma_{\ne}\},0\bigr)$.
\end{theorem}

\begin{proof}[Skeched Proof]
Each of the $u$ differing coordinates contributes at most
$\gamma_{\ne}$. On the remaining $K_Q-u$ coordinates, the clean bits agree, but replacing a shared permanent state with an independent state changes cross-submission correlation by at most $\gamma_{=}$. Independence across coordinates and composition give $\xi_{\mathrm{SP}}(W_0,W_1)$. Search trajectories, returned results, and cross-shard correspondences are
post-processing of the report sequences and fixed index and incur no additional privacy cost.
\end{proof}

The generally nonzero $(K_Q-u)\gamma_{=}$ term is necessary because search-pattern privacy protects the query-reuse relation, not only the query value. Thus, $\xi_{\mathrm{SP}}$ is an XDP parameter over the declared
workload adjacency but cannot be expressed solely using $D_Q$. Since $\gamma_{=}$ and $\gamma_{\ne}$ are calculated directly from PRR/IRR, no separate search-pattern experiment is required. The theorem assumes that no stable query identifier or memoization key is exposed.

\subsubsection{Discussion}

The stored-index theorem quantifies stored-data privacy. The access-pattern
theorem bounds leakage from visible HNSW traversals, while the search-pattern
theorem quantifies repeated-query leakage. The stored-index theorem
conservatively assumes perfect cross-shard association.
Appendix~\ref{app:empirical-leakage-evaluation} evaluates whether geometry-,
topology-, and graph-signal-based methods can recover such correspondences in
practice.


\section{Evaluation}\label{sec:evaluation}

In this section, we evaluate the performance of \name{}. We structure the results around two research questions: 
\begin{enumerate}[leftmargin=*]
    \item How does \name~compare with the baselines in search quality, index construction cost, latency, and communication cost?
    \item How does \name{} scale to support a hundred-million-vector index?
\end{enumerate}

\smallskip\noindent\textbf{Baselines.}
We compare \name~against four baselines.
\begin{itemize}[leftmargin=*]
\item \noindent{\textbf{Plaintext-HNSW.}} This baseline performs non-private semantic search. It builds an HNSW index over the original embeddings using
\texttt{hnswlib} and answers queries directly in plaintext.  We use this to evaluate the privacy overhead of \name{}.

\item \noindent{\textbf{Compass.}}
This baseline combines HNSW search with ORAM~\cite{zhu2025compass}. It hides the logical graph nodes accessed during adaptive search, but requires multiple encrypted block transfers and
client--server interactions. We use the Compass configurations corresponding
to the evaluated datasets and apply the same network settings to its
reported communication cost.

\item \noindent{\textbf{HE-Cluster.}}
This baseline adopts homomorphic encryption for semantic search. It is a variant of Tiptoe~\cite{henzinger2023private} and is also used as a baseline in Compass~\cite{zhu2025compass}. It partitions \(N\) items into approximately
\(\sqrt{N}\) padded semantic clusters, computes encrypted similarities within
the selected cluster, and returns encrypted scores for client-side ranking.

\item \noindent{\textbf{No-Noise.}}
This baseline is the same as \name{} but without storage- and query-side LSHRR. We use this for ablation: the gap between Plaintext-HNSW and
No-Noise captures the effect of binary multi-graph retrieval, while the gap
between No-Noise and \name~captures the impact of randomized
perturbation.
\end{itemize}

\smallskip\noindent\textbf{Datasets and metrics.}
We evaluate SIFT, LAION, TripClick, and MS MARCO, covering dense-vector,
multimodal, and text-retrieval workloads. For SIFT and LAION, we use
Recall@\(K\), which measures the overlap between the returned and ground-truth
Top-\(K\) results. For TripClick and MS MARCO, we use MRR@10, which emphasizes
highly ranked relevant results. Communication is measured by total request and
response bytes and the number of messages. Dataset-specific parameters, query
counts, and candidate budgets are reported in
Appendix~\ref{app:experimental-parameters}.

\smallskip\noindent\textbf{Experiment settings.}
All experiments run on a dual-socket machine with Hygon C86 7375 32-core processors at 2.0\,GHz and 512\,GB of memory. We use Linux traffic control to emulate a LAN with 3\,Gbps bandwidth and 1\,ms RTT and a WAN with 400\,Mbps bandwidth and 80\,ms RTT. 

\subsection{Performance Against Baselines}

\begin{figure*}
\centering
\caption{Search quality versus latency. Hollow and filled markers indicate LAN
and WAN; TripClick and MS MARCO report MRR@10. HE-Cluster is included only
where reported by Compass and is not always quality-matched.}
\label{fig:recall-latency-four}

\begin{adjustbox}{max width=\textwidth,center}
\begin{tikzpicture}
\begin{groupplot}[
    group style={
        group size=4 by 1,
        horizontal sep=0.7cm
    },
    scale only axis,
    width=0.228\textwidth,
    height=0.1283\textwidth,
    xmin=0.01,
    xmax=1200,
    xmode=log,
    log basis x={10},
    xlabel={Latency (ms)},
    xtick={0.01,0.1,1,10,100,1000,10000,100000},
    grid=major,
    grid style={gray!25, line width=0.25pt},
    tick label style={font=\footnotesize},
    label style={font=\footnotesize},
    title style={font=\small},
    every axis plot/.append style={only marks},
]

\nextgroupplot[
    title={Sift},
    ylabel={Recall},
    ymin=0.4,
    ymax=1.02,
    ytick={0.4,0.5,0.6,0.7,0.8,0.9,1.0},
    xmode=log,
    xmin=1e-2,
    xmax=1e5,
    xtick={1e-2,1e-1,1,1e1,1e2,1e3,1e4,1e5}
]

\addplot[
    mark=*,
    mark size=2.0pt,
    color=PaperBlue,
    mark options={fill=white, draw=PaperBlue, line width=0.75pt}
]
coordinates {(0.03,0.98)};

\addplot[
    mark=square*,
    mark size=2.0pt,
    color=PaperOrange,
    mark options={fill=white, draw=PaperOrange, line width=0.75pt}
]
coordinates {(3.40,0.967)};

\addplot[
    mark=triangle*,
    mark size=2.1pt,
    color=PaperGreen,
    mark options={fill=white, draw=PaperGreen, line width=0.75pt}
]
coordinates {(65.02,0.947)};

\addplot[
    mark=diamond*,
    mark size=2.1pt,
    color=PaperRed,
    mark options={fill=white, draw=PaperRed, line width=0.75pt}
]
coordinates {(9.09,0.9583)};

\addplot[
    mark=pentagon*,
    mark size=2.1pt,
    color=PaperPurple,
    mark options={fill=white, draw=PaperPurple, line width=0.75pt}
]
coordinates {(67939.52,0.47)};

\addplot[
    mark=*,
    mark size=2.0pt,
    color=PaperBlue,
    mark options={fill=PaperBlue, draw=PaperBlue, line width=0.75pt}
]
coordinates {(0.85,0.987)};

\addplot[
    mark=square*,
    mark size=2.0pt,
    color=PaperOrange,
    mark options={fill=PaperOrange, draw=PaperOrange, line width=0.75pt}
]
coordinates {(82.79,0.9738)};

\addplot[
    mark=triangle*,
    mark size=2.1pt,
    color=PaperGreen,
    mark options={fill=PaperGreen, draw=PaperGreen, line width=0.75pt}
]
coordinates {(961,0.947)};

\addplot[
    mark=diamond*,
    mark size=2.1pt,
    color=PaperRed,
    mark options={fill=PaperRed, draw=PaperRed, line width=0.75pt}
]
coordinates {(326.11,0.9573)};

\addplot[
    mark=pentagon*,
    mark size=2.1pt,
    color=PaperPurple,
    mark options={fill=PaperPurple, draw=PaperPurple, line width=0.75pt}
]
coordinates {(87421.42,0.47)};

\nextgroupplot[
    title={Laion},
    ymin=0.6,
    ymax=1.02,
    ytick={0.6,0.7,0.8,0.9,1.0},
    xmode=log,
    xmin=1e-2,
    xmax=1e5,
    xtick={1e-2,1e-1,1,1e1,1e2,1e3,1e4,1e5}
    ]

\addplot[
    mark=*,
    mark size=2.0pt,
    color=PaperBlue,
    mark options={fill=white, draw=PaperBlue, line width=0.75pt}
]
coordinates {(0.07,0.9854)};

\addplot[
    mark=square*,
    mark size=2.0pt,
    color=PaperOrange,
    mark options={fill=white, draw=PaperOrange, line width=0.75pt}
]
coordinates {(4.86,0.9756)};

\addplot[
    mark=triangle*,
    mark size=2.1pt,
    color=PaperGreen,
    mark options={fill=white, draw=PaperGreen, line width=0.75pt}
]
coordinates {(42.93,0.974)};

\addplot[
    mark=diamond*,
    mark size=2.1pt,
    color=PaperRed,
    mark options={fill=white, draw=PaperRed, line width=0.75pt}
]
coordinates {(10.81,0.965)};

\addplot[
    mark=pentagon*,
    mark size=2.1pt,
    color=PaperPurple,
    mark options={fill=white, draw=PaperPurple, line width=0.75pt}
]
coordinates {(20238.37,0.9897)};

\addplot[
    mark=*,
    mark size=2.0pt,
    color=PaperBlue,
    mark options={fill=PaperBlue, draw=PaperBlue, line width=0.75pt}
]
coordinates {(4.11,0.9854)};

\addplot[
    mark=square*,
    mark size=2.0pt,
    color=PaperOrange,
    mark options={fill=PaperOrange, draw=PaperOrange, line width=0.75pt}
]
coordinates {(165.50,0.9756)};

\addplot[
    mark=triangle*,
    mark size=2.1pt,
    color=PaperGreen,
    mark options={fill=PaperGreen, draw=PaperGreen, line width=0.75pt}
]
coordinates {(1082.58,0.974)};

\addplot[
    mark=diamond*,
    mark size=2.1pt,
    color=PaperRed,
    mark options={fill=PaperRed, draw=PaperRed, line width=0.75pt}
]
coordinates {(565.81,0.965)};

\addplot[
    mark=pentagon*,
    mark size=2.1pt,
    color=PaperPurple,
    mark options={fill=PaperPurple, draw=PaperPurple, line width=0.75pt}
]
coordinates {(26665.74, 0.9897)};

\nextgroupplot[
    title={TripClick},
    ymin=0.12,
    ymax=0.32,
    ytick={0.1,0.15,0.20,0.25,0.30,0.35},
    xmode=log,
    xmin=1e-2,
    xmax=1e6,
    xtick={1e-2,1e-1,1,1e1,1e2,1e3,1e4,1e5,1e6}
]

\addplot[
    mark=*,
    mark size=2.0pt,
    color=PaperBlue,
    mark options={fill=white, draw=PaperBlue, line width=0.75pt}
]
coordinates {(0.1885,0.2934)};

\addplot[
    mark=square*,
    mark size=2.0pt,
    color=PaperOrange,
    mark options={fill=white, draw=PaperOrange, line width=0.75pt}
]
coordinates {(24.24,0.2925)};

\addplot[
    mark=triangle*,
    mark size=2.1pt,
    color=PaperGreen,
    mark options={fill=white, draw=PaperGreen, line width=0.75pt}
]
coordinates {(340.71,0.2909)};

\addplot[
    mark=diamond*,
    mark size=2.1pt,
    color=PaperRed,
    mark options={fill=white, draw=PaperRed, line width=0.75pt}
]
coordinates {(22.60,0.2898)};

\addplot[
    mark=pentagon*,
    mark size=2.1pt,
    color=PaperPurple,
    mark options={fill=white, draw=PaperPurple, line width=0.75pt}
]
coordinates {(373157.00,0.14)};

\addplot[
    mark=*,
    mark size=2.0pt,
    color=PaperBlue,
    mark options={fill=PaperBlue, draw=PaperBlue, line width=0.75pt}
]
coordinates {(5.0435,0.2934)};

\addplot[
    mark=square*,
    mark size=2.0pt,
    color=PaperOrange,
    mark options={fill=PaperOrange, draw=PaperOrange, line width=0.75pt}
]
coordinates {(1211.10,0.2925)};

\addplot[
    mark=triangle*,
    mark size=2.1pt,
    color=PaperGreen,
    mark options={fill=PaperGreen, draw=PaperGreen, line width=0.75pt}
]
coordinates {(5041.67,0.2909)};

\addplot[
    mark=diamond*,
    mark size=2.1pt,
    color=PaperRed,
    mark options={fill=PaperRed, draw=PaperRed, line width=0.75pt}
]
coordinates {(3571.19,0.2898)};

\addplot[
    mark=pentagon*,
    mark size=2.1pt,
    color=PaperPurple,
    mark options={fill=PaperPurple, draw=PaperPurple, line width=0.75pt}
]
coordinates {(363590.00,0.14)};

\nextgroupplot[
    title={MS MARCO},
    xmax=10000,
    xtick={0.01,0.1,1,10,100,1000,10000},
    ymin=0.30,
    ymax=0.42,
    ytick={0.30,0.34,0.38,0.42}
]

\addplot[
    mark=*,
    mark size=2.0pt,
    color=PaperBlue,
    mark options={fill=white, draw=PaperBlue, line width=0.75pt}
]
coordinates {(0.4119,0.3805)};

\addplot[
    mark=square*,
    mark size=2.0pt,
    color=PaperOrange,
    mark options={fill=white, draw=PaperOrange, line width=0.75pt}
]
coordinates {(36.87,0.3667)};

\addplot[
    mark=triangle*,
    mark size=2.1pt,
    color=PaperGreen,
    mark options={fill=white, draw=PaperGreen, line width=0.75pt}
]
coordinates {(607.75,0.3560)};

\addplot[
    mark=diamond*,
    mark size=2.1pt,
    color=PaperRed,
    mark options={fill=white, draw=PaperRed, line width=0.75pt}
]
coordinates {(130.29,0.3617)};

\addplot[
    mark=*,
    mark size=2.0pt,
    color=PaperBlue,
    mark options={fill=PaperBlue, draw=PaperBlue, line width=0.75pt}
]
coordinates {(5.2708,0.3805)};

\addplot[
    mark=square*,
    mark size=2.0pt,
    color=PaperOrange,
    mark options={fill=PaperOrange, draw=PaperOrange, line width=0.75pt}
]
coordinates {(1801.31,0.3667)};

\addplot[
    mark=triangle*,
    mark size=2.1pt,
    color=PaperGreen,
    mark options={fill=PaperGreen, draw=PaperGreen, line width=0.75pt}
]
coordinates {(7763.51,0.3560)};

\addplot[
    mark=diamond*,
    mark size=2.1pt,
    color=PaperRed,
    mark options={fill=PaperRed, draw=PaperRed, line width=0.75pt}
]
coordinates {(6002.67,0.3617)};

\end{groupplot}

\node[
    anchor=south,
    inner sep=0pt,
    text width=0.98\textwidth,
    align=center
] at ($(group c2r1.north)!0.5!(group c3r1.north)+(0,0.84cm)$) {%
{\footnotesize
\setlength{\tabcolsep}{2pt}
\begin{tabular*}{\linewidth}{@{\extracolsep{\fill}}lccccc@{}}
{\(\circ\) LAN;\ \(\bullet\) WAN}
&
\tikz[baseline=-0.5ex]{\draw[PaperBlue, line width=0.75pt, fill=white]
(0,0) circle (1.8pt);}
\hspace{-1pt}
\tikz[baseline=-0.5ex]{\draw[PaperBlue, line width=0.75pt, fill=PaperBlue]
(0,0) circle (1.8pt);}
~HNSW
&
\tikz[baseline=-0.5ex]{\draw[PaperOrange, line width=0.75pt, fill=white]
(-1.8pt,-1.8pt) rectangle (1.8pt,1.8pt);}
\hspace{-1pt}
\tikz[baseline=-0.5ex]{\draw[PaperOrange, line width=0.75pt, fill=PaperOrange]
(-1.8pt,-1.8pt) rectangle (1.8pt,1.8pt);}
~No-Noise
&
\tikz[baseline=-0.5ex]{\draw[PaperPurple, line width=0.75pt, fill=white]
(90:2.2pt)--(18:2.2pt)--(-54:2.2pt)--(-126:2.2pt)--(162:2.2pt)--cycle;}
\hspace{-1pt}
\tikz[baseline=-0.5ex]{\draw[PaperPurple, line width=0.75pt, fill=PaperPurple]
(90:2.2pt)--(18:2.2pt)--(-54:2.2pt)--(-126:2.2pt)--(162:2.2pt)--cycle;}
~HE-Cluster
&
\tikz[baseline=-0.5ex]{\draw[PaperGreen, line width=0.75pt, fill=white]
(90:2.35pt)--(210:2.35pt)--(330:2.35pt)--cycle;}
\hspace{-1pt}
\tikz[baseline=-0.5ex]{\draw[PaperGreen, line width=0.75pt, fill=PaperGreen]
(90:2.35pt)--(210:2.35pt)--(330:2.35pt)--cycle;}
~Compass
&
\tikz[baseline=-0.5ex]{\draw[PaperRed, line width=0.75pt, fill=white]
(0,2.45pt)--(2.45pt,0)--(0,-2.45pt)--(-2.45pt,0)--cycle;}
\hspace{-1pt}
\tikz[baseline=-0.5ex]{\draw[PaperRed, line width=0.75pt, fill=PaperRed]
(0,2.45pt)--(2.45pt,0)--(0,-2.45pt)--(-2.45pt,0)--cycle;}
~MESS
\end{tabular*}
}
};

\end{tikzpicture}
\end{adjustbox}
\end{figure*}

\begin{table*}
\centering
	\caption{Latency and communication cost. HNSW is non-private; No-Noise disables
perturbation; HE-Cluster and Compass use HE and ORAM, respectively; Ours denotes
\name. Two messages form one request--response round.}
	\label{comprehensiveeval}
	\resizebox{\textwidth}{!}{%
		\begin{tabular}{@{}ll|c|ccccc|c|c@{}}
			\toprule
			\multirow{2}{*}{\textbf{Dataset}} 
			& \multirow{2}{*}{\textbf{Scheme}} 
			& \multicolumn{1}{c|}{\textbf{Offline Setup Phase}} 
			& \multicolumn{5}{c|}{\textbf{Online Search Latency (ms/query)}} 
            & \multirow{2}{*}{\textbf{\begin{tabular}[c]{@{}c@{}}Messages\\ per Query\end{tabular}}}
			& \multirow{2}{*}{\textbf{\begin{tabular}[c]{@{}c@{}}Comm. Overhead\\ (KB/query)\end{tabular}}} \\ 
			\cmidrule(lr){3-3} 
			\cmidrule(lr){4-8}
			&  
			& \textbf{Build Time (s)} 
			& \textbf{Mode}
			& \textbf{Server Eval.} 
			& \textbf{Trans. Lat.}
			& \textbf{Client Eval.} 
			& \textbf{Total Time} 
			& \\ 
			\midrule
			
			\multirow{8}{*}{\begin{tabular}[c]{@{}l@{}}\textbf{SIFT}\\ $(10\mathrm{K},128d)$\end{tabular}} 
			& \multirow{2}{*}{HNSW \cite{malkov2018efficient}} 
			& \multirow{2}{*}{$44.42$} 
			& WAN & $0.02$ & $0.83$ & -- & $0.85$ & \multirow{2}{*}{$2$} & \multirow{2}{*}{$0.64$} \\
			&  &  & LAN & $0.02$ & $0.01$ & -- & $0.03$ & \\

			& \multirow{2}{*}{No-Noise } 
			& \multirow{2}{*}{$435.40$} 
			& WAN & $1.08$ & $80.73$ & $0.97$ & $82.78$ & \multirow{2}{*}{$2$} & \multirow{2}{*}{$138.07$} \\
			&  &  & LAN & $0.99$ & $1.52$ & $0.89$ & $3.40$ & \\

            & \multirow{2}{*}{{HE-Cluster}} 
			& \multirow{2}{*}{$-$} 
			& WAN & $55621.65$ & $24083.69$ & $7716.07$ & $87421.42$ & \multirow{2}{*}{$2$} & \multirow{2}{*}{$744096.35$} \\
            &  &  & LAN & $55980.91$ & $4250.11$ & $7708.48$ & $67939.52$ & \\

			& \multirow{2}{*}{\textbf{Compass} \cite{zhu2025compass}} 
			& \multirow{2}{*}{$\ge 44.42$} 
			& WAN & $10.03$ & $936.62$ & $11.70$ & $958.35$ & \multirow{2}{*}{$16.1048$} & \multirow{2}{*}{$11451.86$} \\
			&  &  & LAN & $13.74$ & $36.85$ & $14.43$ & $65.02$ & \\

			& \multirow{2}{*}{\textbf{Ours}} 
			& \multirow{2}{*}{$512.74$} 
			& WAN & $1.45$ & $323.32$ & $1.34$ & $\mathbf{326.11}$ & \multirow{2}{*}{$2$} & \multirow{2}{*}{$788.82$} \\
			&  &  & LAN & $1.16$ & $6.75$ & $1.18$ & $\mathbf{9.09}$ & \\
            
			\midrule

			\multirow{8}{*}{\begin{tabular}[c]{@{}l@{}}\textbf{LAION}\\ $(1\mathrm{M},512d)$\end{tabular}} 
			& \multirow{2}{*}{HNSW \cite{malkov2018efficient}} 
			& \multirow{2}{*}{$6.53$} 
			& WAN & $0.02$ & $4.11$ & -- & $4.13$ & \multirow{2}{*}{$2$} & \multirow{2}{*}{$2.14$} \\
			&  &  & LAN & $0.02$ & $0.07$ & -- & $0.09$ & \\

			& \multirow{2}{*}{No-Noise } 
			& \multirow{2}{*}{$25.64$} 
			& WAN & $1.20$ & $163 .11$ & $1.19$ & $165.50$ & \multirow{2}{*}{$2$} & \multirow{2}{*}{$391.82$} \\
			&  &  & LAN & $1.29$ & $2.54$ & $1.02$ & $4.85$ & \\

            & \multirow{2}{*}{{HE-Cluster}} 
			& \multirow{2}{*}{$-$} 
			& WAN & $16408.31$ & $8040.37$ & $2217.09$ & $26665.74$ & \multirow{2}{*}{$2$} & \multirow{2}{*}{$208941.03$} \\
            &  &  & LAN & $16408.31$ & $1601.69$ & $2228.37$ & $20238.37$ & \\

			& \multirow{2}{*}{\textbf{Compass} \cite{zhu2025compass}} 
			& \multirow{2}{*}{$\ge 6.53$} 
			& WAN & $2.74$ & $1068.67$ & $11.18$ & $1082.59$ & \multirow{2}{*}{$16.002$} & \multirow{2}{*}{$45989.36$} \\
			&  &  & LAN & $3.12$ & $28.35$ & $11.45$ & $42.92$ & \\

			& \multirow{2}{*}{\textbf{Ours}} 
			& \multirow{2}{*}{$44.33$} 
			& WAN & $1.12$ & $563.55$ & $1.14$ & $\mathbf{565.81}$ & \multirow{2}{*}{$2$} & \multirow{2}{*}{$1303.57$} \\
			&  &  & LAN & $0.92$ & $9.20$ & $0.70$ & $\mathbf{10.82}$ & \\
			\midrule

			\multirow{8}{*}{\begin{tabular}[c]{@{}l@{}}\textbf{TripClick}\\ $(1.5\mathrm{M},768d)$\end{tabular}} 
			& \multirow{2}{*}{HNSW \cite{malkov2018efficient}} 
			& \multirow{2}{*}{$1211.97$} 
			& WAN & $0.13$ & $4.92$ & -- & $5.05$ & \multirow{2}{*}{$2$} & \multirow{2}{*}{$3.14$} \\
			&  &  & LAN & $0.07$ & $0.12$ & -- & $0.19$ & \\

			& \multirow{2}{*}{No-Noise } 
			& \multirow{2}{*}{$1137.13$} 
			& WAN & $2.88$ & $1206.90$ & $1.32$ & $1211.10$ & \multirow{2}{*}{$2$} & \multirow{2}{*}{$2914.82$} \\
			&  &  & LAN & $2.33$ & $20.82$ & $1.10$ & $24.25$ & \\

            & \multirow{2}{*}{{HE-Cluster}} 
			& \multirow{2}{*}{$-$} 
			& WAN  & $325800.00$ 
            & $29180.00$ 
            & $8620.00$ 
            & $363590.00$ 
            & \multirow{2}{*}{$2$} & \multirow{2}{*}{$782504.63$} \\
            &  &  & LAN & $358757.20$ 
            & $5502.80$ 
            & $8897.02$ 
            & $373157.00$\\

			& \multirow{2}{*}{\textbf{Compass} \cite{zhu2025compass}} 
			& \multirow{2}{*}{$\ge 1211.97$} 
			& WAN & $134.82$ & $4774.15$ & $132.70$ & $5041.67$ & \multirow{2}{*}{$18.0613$} & \multirow{2}{*}{$133042.91$} \\
			&  &  & LAN & $25.45$ & $200.60$ & $114.65$ & $340.70$ & \\

			& \multirow{2}{*}{\textbf{Ours}} 
			& \multirow{2}{*}{$1586.19$} 
			& WAN & $6.46$ & $3563.14$ & $1.59$ & $\mathbf{3571.19}$ & \multirow{2}{*}{$2$} & \multirow{2}{*}{$8936.57$} \\
			&  &  & LAN & $2.97$ & $18.01$ & $1.62$ & $\mathbf{22.60}$ & \\
			\midrule

			\multirow{8}{*}{\begin{tabular}[c]{@{}l@{}}\textbf{MS MARCO}\\ $(8.8\mathrm{M},768d)$\end{tabular}} 
			& \multirow{2}{*}{HNSW \cite{malkov2018efficient}} 
			& \multirow{2}{*}{$10728.84$} 
			& WAN & $0.22$ & $5.05$ & -- & $5.27$ & \multirow{2}{*}{$2$} & \multirow{2}{*}{$3.15$} \\
			&  &  & LAN & $0.18$ & $0.23$ & -- & $0.41$ & \\

			& \multirow{2}{*}{No-Noise } 
			& \multirow{2}{*}{$8676.12$} 
			& WAN & $4.33$ & $1795.45$ & $1.53$ & $1801.31$ & \multirow{2}{*}{$2$} & \multirow{2}{*}{$4468.82$} \\
			&  &  & LAN & $3.03$ & $32.33$ & $1.52$ & $36.88$ & \\

            & \multirow{2}{*}{{HE-Cluster}} 
			& \multirow{2}{*}{$-$} 
			& WAN  & $-$ & $-$ & $-$ & $-$ & \multirow{2}{*}{$-$} & \multirow{2}{*}{$-$}\\
            &  &  & LAN & $-$ & $-$ & $-$ & $-$ & \\

			& \multirow{2}{*}{\textbf{Compass} \cite{zhu2025compass}} 
			& \multirow{2}{*}{$\ge 10728.84$} 
			& WAN & $122.32$ & $7445.57$ & $195.61$ & $7763.50$ & \multirow{2}{*}{$18.0209$} & \multirow{2}{*}{$226657.48$} \\
			&  &  & LAN & $64.01$ & $361.85$ & $181.90$ & $607.76$ & \\

			& \multirow{2}{*}{\textbf{Ours}} 
			& \multirow{2}{*}{$11741.53$} 
			& WAN & $4.20$ & $5997.02$ & $1.45$ & $\mathbf{6002.67}$ & \multirow{2}{*}{$2$} & \multirow{2}{*}{$18066.32$} \\
			&  &  & LAN & $6.98$ & $120.47$ & $2.84$ & $\mathbf{130.29}$ & \\ 
           
			\bottomrule
		\end{tabular}%
	}
\end{table*}

\autoref{fig:recall-latency-four} compares the search quality of \name{} and the baselines. \autoref{comprehensiveeval} breaks down the latency and communication cost. 
For graph-based methods, we report results on configurations selected for high accuracy. The results for HE-Cluster are based on the configurations used in Compass~\cite{zhu2025compass}. 

\smallskip\noindent\textbf{Search quality.}
\autoref{fig:recall-latency-four} shows that \name~remains close to
Plaintext-HNSW: the Recall gap is below \(0.05\), and the MRR gap is below
\(0.03\). We can explain this gap using No-Noise's result. Specifically, although No-Noise removes randomized perturbation, it still searches finite-length binary embeddings in Hamming space, distributes items over several graph shards, aggregates candidates, transfers encrypted payloads, removes
duplicates, and reranks on the client. Its gap to Plaintext-HNSW represents the effect of searching binary multi-graph without privacy noise. Its gap to \name~is attributable to LSHRR and the larger candidate set needed to compensate for randomized bit flips.
We can see that adding LSHRR leads to a quality change of at most 0.011. This small quality change comes at the cost of efficiency, because \name~must retrieve
more candidates to maintain quality, increasing communication cost. 


\smallskip\noindent\textbf{Comparison with cryptographic baselines.}
Compass achieves accuracy close to that of non-private search, but ORAM-protected
adaptive traversal requires repeated accesses to encrypted blocks. \autoref{comprehensiveeval} shows that \name~is \(4.0\)--\(15.1\times\) faster under LAN and \(1.3\)--\(2.9\times\) faster under WAN. The largest LAN improvement appears on TripClick, where the cost of interactive traversal is especially high.
The WAN advantage is smaller because transfer time dominates in both systems as network speed decreases. HE-Cluster shows a different bottleneck. On the three datasets for which Compass reports results, encrypted similarity evaluation and encrypted score transfer require approximately \(20\)--\(373\) seconds and
\(209\)--\(783\) MB per query. Although its quality is competitive on LAION, the results indicate significant computational and communication overhead that is orders of magnitude higher than \name{}. In summary, the two baselines that adopt cryptographic approaches suffer from high overhead: HE-Cluster incurs significant overhead in vector computation and produces large ciphertexts, and Compass incurs large communication overhead.

\begin{figure}
\centering

\begin{tikzpicture}

\begin{axis}[
    width=\linewidth,
    height=0.5625\linewidth,
    ybar,
    bar width=7pt,
    xmin=-0.55,
    xmax=3.55,
    ymin=0,
    ymax=115,
    ytick={0,20,40,60,80,100},
    ylabel={Recall (\%)},
    axis y line*=left,
    axis x line*=bottom,
    xtick={0,1,2,3},
    xticklabels={
        {SIFT\\Pool=1500},
        {LAION\\Pool=700},
        {TripClick\\Pool=3000},
        {MS MARCO\\Pool=7000}
    },
    xticklabel style={
        align=center,
        font=\scriptsize
    },
    tick label style={
        font=\scriptsize
    },
    label style={
        font=\scriptsize
    },
    grid=major,
    grid style={
        gray!25,
        line width=0.25pt
    },
    enlarge x limits=false,
    clip=false,
    legend style={
        font=\scriptsize,
        draw=none,
        fill=white,
        fill opacity=0.9,
        text opacity=1,
        at={(0.5,0.98)},
        anchor=north,
        legend columns=4,
        column sep=4pt,
        inner xsep=2pt,
        inner ysep=1pt
    },
    extra description/.code={
        \draw[
            black,
            line width=0.4pt
        ]
        (rel axis cs:0,1) -- (rel axis cs:1,1);
    },
]

\addplot[
    draw=PaperBlue,
    fill=PaperBlue!45,
    line width=0.6pt,
    bar shift=-10.5pt
]
coordinates {
    (0,57.0)
    (1,57.9)
};
\addlegendentry{$t=1$}

\addplot[
    draw=PaperOrange,
    fill=PaperOrange!50,
    line width=0.6pt,
    bar shift=-3.5pt
]
coordinates {
    (0,77.2)
    (1,76.2)
};
\addlegendentry{$t=4$}

\addplot[
    draw=PaperGreen,
    fill=PaperGreen!50,
    line width=0.6pt,
    bar shift=3.5pt
]
coordinates {
    (0,84.7)
    (1,83.7)
};
\addlegendentry{$t=8$}

\addplot[
    draw=PaperRed,
    fill=PaperRed!45,
    line width=0.6pt,
    bar shift=10.5pt
]
coordinates {
    (0,95.73)
    (1,96.5)
};
\addlegendentry{$t=16$}

\end{axis}

\begin{axis}[
    width=\linewidth,
    height=0.5625\linewidth,
    ybar,
    bar width=7pt,
    xmin=-0.55,
    xmax=3.55,
    ymin=0,
    ymax=0.44,
    ytick={0,0.1,0.2,0.3,0.4},
    ylabel={MRR},
    axis y line*=right,
    axis x line=none,
    xtick=\empty,
    tick label style={
        font=\scriptsize
    },
    label style={
        font=\scriptsize
    },
    enlarge x limits=false,
    clip=false,
]

\addplot[
    draw=PaperBlue,
    fill=PaperBlue!45,
    line width=0.6pt,
    bar shift=-10.5pt
]
coordinates {
    (2,0.249885)
    (3,0.27628)
};

\addplot[
    draw=PaperOrange,
    fill=PaperOrange!50,
    line width=0.6pt,
    bar shift=-3.5pt
]
coordinates {
    (2,0.259944)
    (3,0.307478)
};

\addplot[
    draw=PaperGreen,
    fill=PaperGreen!50,
    line width=0.6pt,
    bar shift=3.5pt
]
coordinates {
    (2,0.265512)
    (3,0.321746)
};

\addplot[
    draw=PaperRed,
    fill=PaperRed!45,
    line width=0.6pt,
    bar shift=10.5pt
]
coordinates {
    (2,0.2898)
    (3,0.3617)
};

\end{axis}

\end{tikzpicture}

\caption{Effect of the number of selected shards \(t\) on search quality
with \(M=64\) and a fixed candidate pool for each dataset. SIFT and LAION are
evaluated using Recall, while TripClick and MS MARCO are evaluated using MRR.
}
\label{fig:t_retrieval_quality_single}

\end{figure}
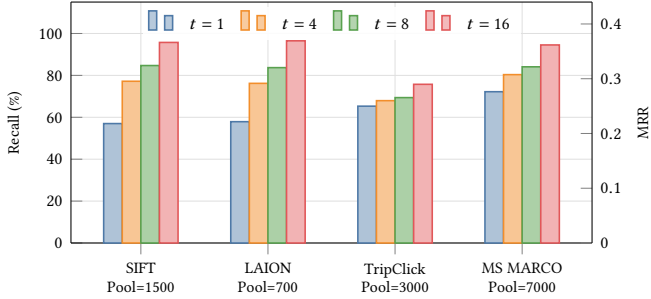

\begin{table*}
\centering

\caption{SIFT100M performance under different query-side perturbation
parameters.}
\label{tab:sift100m_eval}

\scriptsize
\setlength{\tabcolsep}{4.2pt}
\renewcommand{\arraystretch}{0.95}

\begin{adjustbox}{width=0.86\textwidth,center}
\begin{tabular}{@{}lc|c|c|cccc|cc@{}}
\toprule
\multirow{2}{*}{\textbf{Dataset}}
& \multirow{2}{*}{\textbf{Build Time}}
& \multirow{2}{*}{\textbf{$p$}}
& \multirow{2}{*}{\textbf{Mode}}
& \multicolumn{4}{c|}{\textbf{Online Search Performance}}
& \multirow{2}{*}{\textbf{Candidate}}
& \multirow{2}{*}{\textbf{Communication}} \\
\cmidrule(lr){5-8}
& & &
& \textbf{Server Eval.}
& \textbf{Trans.Lat.}
& \textbf{Client Dec.\&Rank}
& \textbf{Total Time}
& & \\
\midrule

\multirow{8}{*}{\textbf{SIFT100M}}
& \multirow{8}{*}{26.19 h}
& \multirow{2}{*}{$0.00$}
& WAN
& 3.21 & 888.89 & 1.63 & 893.73 & 4000 & 2123.54 \\

& & & LAN
& 1.98 & 21.52 & 1.36 & 24.86 & 4000 & 2123.54 \\

\cmidrule(lr){3-10}

& & \multirow{2}{*}{$0.08$}
& WAN
& 8.61 & 2031.79 & 2.45 & 2042.85 & 9500 & 5069.04 \\

& & & LAN
& 3.43 & 46.91 & 2.19 & 52.53 & 9500 & 5069.04 \\

\cmidrule(lr){3-10}

& & \multirow{2}{*}{$0.10$}
& WAN
& 14.33 & 2666.87 & 2.86 & 2684.06 & 12500 & 6678.79 \\

& & & LAN
& 11.40 & 57.53 & 2.55 & 71.48 & 12500 & 6678.79 \\

\cmidrule(lr){3-10}

& & \multirow{2}{*}{$0.12$}
& WAN
& 41.78 & 4225.49 & 3.90 & 4271.17 & 20000 & 10686.04 \\

& & & LAN
& 10.99 & 100.49 & 5.45 & 116.93 & 20000 & 10686.04 \\

\bottomrule
\end{tabular}
\end{adjustbox}

\end{table*}

{
\smallskip\noindent\textbf{Impact of $t$. }
\autoref{fig:t_retrieval_quality_single} evaluates the impact of $t$ on search quality, when $M$ is fixed at $64$. It can be seen that search 
quality improves as more shards are selected. We also note that larger \(t\) also increases the number of
shard-local entries and the associated storage cost. Increasing \(t\)
from 1 to 16 raises Recall from \(57.0\%\) to \(95.73\%\) on SIFT and from
\(57.9\%\) to \(96.5\%\) on LAION. TripClick MRR increases from \(0.250\) to
\(0.290\), while MS MARCO MRR increases from \(0.276\) to \(0.362\). This is because a larger
\(t\) places each vector in more independently perturbed graph views, increasing
the probability that it retains a favorable rank in at least one
shard. The gains are more pronounced for the Recall metric than for the MRR metric, because the latter considers the vector's final position after exact reranking.
}


\smallskip
\noindent\textbf{Latency breakdown.}
For \name~under LAN, server search remains below \(7\) ms and client decryption
and reranking below \(3\) ms across the four datasets. Therefore, both Hamming
graph traversal and local cryptographic processing are efficient. The remaining cost is encrypted candidate transfer, which becomes significant under WAN: on MS MARCO, this transfer accounts for 
\(5997.02\) ms of the \(6002.67\) ms end-to-end latency. Compass is also 
network-bound under WAN, but incurs high computation cost due to ORAM operations. These results suggest that optimizing candidate set and
payload size are interesting future extensions for \name.


\smallskip\noindent\textbf{Offline cost.}
The offline cost of \name~includes shard-specific hashing, independent
perturbation, payload encryption, and construction of multiple HNSW
shards. It is therefore notably higher than constructing one plaintext, as shown in~\autoref{comprehensiveeval}. However, this cost is incurred once and can be amortized over subsequent queries.

\smallskip\noindent\textbf{Insertion and deletion.}
\name~supports incremental updates without rebuilding the complete index. For insertion, the client computes and perturbs the new item's shard-specific binary embeddings, encrypts its payload, and sends it to the \(t=16\) shards. The server then inserts the perturbed item into the shards. Deletion is performed lazily: the client identifies the shard-local labels, and the server marks the corresponding nodes as inactive, thereby excluding them from future searches. On MS MARCO, insertion takes \(96.16\) ms over LAN and
\(351.29\) ms over WAN, while deletion takes \(1.33\) ms and \(80.35\) ms, respectively. Their communication costs are \(3.09\) KB and \(0.09\) KB. 

\begin{figure}
\centering

\begin{subfigure}{\linewidth}
\centering

\begin{tikzpicture}

\begin{axis}[
    width=\linewidth,
    height=0.5625\linewidth,
    ybar,
    bar width=8pt,
    symbolic x coords={0.00,0.08,0.10,0.12},
    xtick=data,
    xlabel={Query Perturbation Parameter $p$},
    ylabel={WAN Latency (ms)},
    ymin=0,
    ymax=4600,
    ytick={0,1000,2000,3000,4000},
    scaled y ticks=false,
    axis y line*=left,
    axis x line*=bottom,
    enlarge x limits=0.22,
    clip=false,
    tick label style={
        font=\scriptsize
    },
    label style={
        font=\scriptsize
    },
    grid=major,
    grid style={
        gray!25,
        line width=0.2pt
    },
    legend style={
        font=\scriptsize,
        draw=none,
        fill=white,
        at={(0.03,0.97)},
        anchor=north west
    },
    extra description/.code={
        \draw[
            black,
            line width=0.4pt
        ]
        (rel axis cs:0,1) -- (rel axis cs:1,1);
    },
]

\addplot[
    bar shift=-5pt,
    draw=PaperBlue,
    fill=PaperBlue!35,
    line width=0.6pt
]
coordinates {
    (0.00,893.73)
    (0.08,2042.85)
    (0.10,2684.06)
    (0.12,4271.17)
};
\addlegendentry{WAN}

\end{axis}

\begin{axis}[
    width=\linewidth,
    height=0.5625\linewidth,
    ybar,
    bar width=8pt,
    symbolic x coords={0.00,0.08,0.10,0.12},
    xtick=\empty,
    ylabel={LAN Latency (ms)},
    ymin=0,
    ymax=130,
    ytick={0,20,40,60,80,100,120},
    scaled y ticks=false,
    axis y line*=right,
    axis x line=none,
    enlarge x limits=0.22,
    clip=false,
    tick label style={
        font=\scriptsize
    },
    label style={
        font=\scriptsize
    },
    legend style={
        font=\scriptsize,
        draw=none,
        fill=white,
        at={(0.03,0.86)},
        anchor=north west
    },
]

\addplot[
    bar shift=5pt,
    draw=PaperOrange,
    fill=PaperOrange!45,
    line width=0.6pt
]
coordinates {
    (0.00,24.86)
    (0.08,52.53)
    (0.10,71.48)
    (0.12,116.93)
};
\addlegendentry{LAN}

\end{axis}

\end{tikzpicture}

\caption{Total search latency.}
\label{fig:sift100m-dual-axis-latency}
\end{subfigure}

\vspace{-1mm}

\begin{subfigure}{\linewidth}
\centering

\begin{tikzpicture}

\begin{axis}[
    width=\linewidth,
    height=0.5625\linewidth,
    xlabel={Candidate Pool ($\times 10^3$)},
    ylabel={Recall (\%)},
    xmin=800,
    xmax=12200,
    ymin=40,
    ymax=100,
    ytick={40,50,60,70,80,90,100},
    scaled x ticks=false,
    scaled y ticks=false,
    axis y line*=left,
    axis x line*=bottom,
    clip=false,
    tick label style={
        font=\scriptsize
    },
    label style={
        font=\scriptsize
    },
    grid=major,
    grid style={
        gray!25,
        line width=0.2pt
    },
    legend style={
        font=\scriptsize,
        draw=none,
        fill=white
    },
    legend pos=south east,
    xtick={
        1000,2000,3000,4000,5000,6000,
        7000,8000,9000,10000,11000,12000
    },
    xticklabels={1,2,3,4,5,6,7,8,9,10,11,12},
    extra description/.code={
        \draw[
            black,
            line width=0.4pt
        ]
        (rel axis cs:0,1) -- (rel axis cs:1,1);
        \draw[
            black,
            line width=0.4pt
        ]
        (rel axis cs:1,0) -- (rel axis cs:1,1);
    },
]

\addplot[
    thick,
    color=PaperBlue,
    mark=*,
    mark size=1.6pt
]
coordinates {
    (1000,61.9)
    (2000,76.0)
    (3000,82.2)
    (4000,86.7)
    (5000,89.4)
    (6000,91.0)
    (7000,92.6)
    (8000,93.2)
    (9000,93.7)
    (10000,94.2)
    (11000,95.0)
    (12000,95.3)
};
\addlegendentry{$p=0.00$}

\addplot[
    thick,
    color=PaperOrange,
    mark=square*,
    mark size=1.6pt
]
coordinates {
    (1000,53.1)
    (2000,64.5)
    (3000,73.0)
    (4000,73.9)
    (5000,77.4)
    (6000,82.3)
    (7000,83.8)
    (8000,86.0)
    (9000,86.9)
    (10000,87.8)
    (11000,88.8)
    (12000,89.3)
};
\addlegendentry{$p=0.08$}

\addplot[
    thick,
    color=PaperGreen,
    mark=triangle*,
    mark size=1.8pt
]
coordinates {
    (1000,47.7)
    (2000,62.1)
    (3000,65.4)
    (4000,69.5)
    (5000,74.1)
    (6000,77.2)
    (7000,81.2)
    (8000,82.5)
    (9000,83.8)
    (10000,84.4)
    (11000,86.4)
    (12000,87.0)
};
\addlegendentry{$p=0.10$}

\addplot[
    thick,
    color=PaperRed,
    mark=diamond*,
    mark size=1.8pt
]
coordinates {
    (1000,42.2)
    (2000,55.5)
    (3000,62.0)
    (4000,67.1)
    (5000,71.1)
    (6000,72.1)
    (7000,77.2)
    (8000,80.0)
    (9000,80.5)
    (10000,82.6)
    (11000,82.8)
    (12000,83.7)
};
\addlegendentry{$p=0.12$}

\end{axis}
\end{tikzpicture}

\caption{Recall under different candidate-pool sizes.}
\label{fig:sift100m-recall-pool}
\end{subfigure}

\caption{SIFT100M performance under different query-side perturbation
parameters: (a) total search latency under LAN and WAN settings; and
(b) Recall@10 under different candidate-pool sizes.}
\label{fig:sift100m-comparison}
\end{figure}
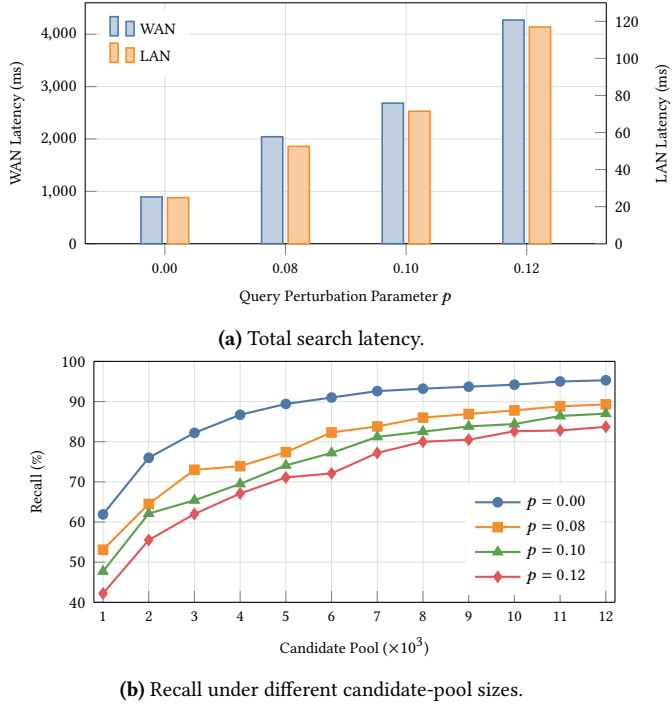

\subsection{Scalability}

\noindent
We evaluate the performance of \name{} at scale, using a large dataset (SIFT100M) created by selecting the first 100 million vectors from SIFT1B. We note that neither  Compass~\cite{zhu2025compass} nor iptoe~\cite{henzinger2023private} report results at this scale. We focus only on the online-phase performance and note that the offline phase takes 26.19 hours. 

\autoref{fig:sift100m-comparison} and
\autoref{tab:sift100m_eval} show the effect of the query-side
perturbation probability \(p\) on Recall@10, candidate volume, latency, and
communication.  With 12,000 candidates, increasing \(p\)
from \(0\) to \(0.12\) reduces Recall@10 from \(95.3\%\) to \(83.7\%\).
It can be seen that the gap decreases with larger candidate sets. However, as the candidate sets grow from 4,000 to 20,000 candidates, the communication cost increases \(5.0\times\). These results demonstrate the accuracy and efficiency trade-off. 
The latency breakdown in \autoref{tab:sift100m_eval} shows that the costs of graph search and client reranking remain
small compared to that of encrypted candidate transfer. In particular, at \(p=0.08\),
\name~processes SIFT100M in \(52.53\) ms/query under LAN; even at \(p=0.12\),
LAN latency remains below \(120\) ms/query. Under WNA, the latency
increases from \(893.73\) ms to \(4271.17\) ms due to large payloads. 
These results confirm that \name{} is practical for large-scale datasets. They highlight the fundamental trade-off: stronger query perturbation increases rank
uncertainty, therefore requires larger candidate sets and incurs higher communication overhead. The results indicate that reducing encrypted payload size or improving shard-level candidate filtering is more beneficial than optimizing the relatively small graph-search and reranking costs.

\section{Related Work}\label{sec:relatedwork}
\noindent\textbf{Vector Databases and ANN Search.}
Vector databases use approximate nearest-neighbor (ANN) indexes to support large-scale, high-dimensional embedding search. HNSW provides low-latency and high-recall graph traversal, while FAISS and SPANN study GPU acceleration and memory--disk hybrid indexing at billion-vector scale \cite{malkov2018efficient,johnson2021billion,chen2021spann}. Recent systems also consider database-level requirements: HAKES supports high-recall search under concurrent updates, while VBase integrates vector similarity search with relational query processing~\cite{hu2025hakes,zhang2023vbase}. These systems
target performance and scalability of plaintext vector search,
whereas \name targets privacy. While indexes other than HNSW can be used in \name{}, our current security analysis assumes the use of graph indexes. Extending \name{} to support other indexes is left as future work. 

\smallskip
\noindent\textbf{Encrypted Semantic Retrieval.}
Existing works reduce secure candidate-generation costs through clustering, LSH, tree indexes, or stronger deployment assumptions. Tiptoe and Wally use clustering, with Wally additionally combining epochs and differential privacy \cite{henzinger2023private,asi2024scalable}. Preco uses LSH with two
non-colluding servers, while VCC24 combines tree indexing and \(k\)-means under information-theoretic security
\cite{servan2022private,vithana2024private}. SANNS uses DORAM and garbled circuits for secure Top-\(K\) selection~\cite{chen2020sanns}, whereas trusted-hardware approaches introduce additional trust assumptions \cite{amjad2019forward,vo2021towards}.
{
Other works adopt Private Information Retrieval (PIR). In particular,  Preco notes that replacing its non-colluding servers with single-server PIR would increase overhead, while Panther combines PIR2A with secret sharing for private access and secure Top-\(K\) selection
\cite{servan2022private,li2025panther}.
}

Cryptographic solutions use order- or property-preserving encryption for distance comparison, but the ciphertext may leak metric
relations~\cite{wong2009secure,Huang2019Toward,liu2025privacy}. Compass uses ORAM to hide HNSW accesses~\cite{zhu2025compass}, while PACMANN combines PIR with a secure neighborhood graph~\cite{zhou2024pacmann}. These solutions provide stronger obliviousness but suffer significant overhead from cryptographic operations. 
In contrast, \name searches server-visible HNSW shards over perturbed binary embeddings, with well-defined stored-data, access-pattern, and search-pattern leakage.

\section{Conclusion}\label{sec:conclusion}
We presented \name{} that enables semantic search with privacy, accuracy, and efficiency. The system adopts a differential privacy approach, allowing search to be performed over perturbed vectors while bounding leakage of stored data, access patterns, and search patterns. \name{} proposes a multi-graph design that compensates for the accuracy loss from perturbation. The evaluation results show that it outperforms cryptographic baselines, reducing communication costs by up to \(35.28\times\) compared to Compass. The system remains practical at scale, achieving 52.53 ms/query on a dataset of 100M vectors. 
\bibliographystyle{ACM-Reference-Format}
\bibliography{crypto/reference}

%

\appendix
\label{app:Leakage}

\appendix

\section{Detailed Accounting and Empirical Leakage Evaluation}
\label{app:detailed-security-analysis}

This appendix complements Section~\ref{sec:security-analysis}. The XDP
definitions, mechanism-level guarantees, and proofs are given in the main
text. Here, we describe the numerical privacy accountants and evaluate how
much multi-shard XDP evidence can be organized by concrete cross-shard
inference procedures.

\subsection{Implementation of the Privacy Accountants}
\label{app:privacy-accountant-implementation}

\subsubsection{Stored-Data Accountant}

For a real challenge pair $q=(\bm x_0,\bm x_1)$, we compute the fixed-ISO-LSH
distance
$d_s(q)=d_H(H_s(\bm x_0),H_s(\bm x_1))$
for every frozen shard mapping. Under uniform $t$-of-$M$ assignment, let
$w_q(u)$ be the probability that the aggregate distance over the selected
shards equals $u$.

We compute $w_q(u)$ by dynamic programming rather than enumerating all
$\binom{M}{t}$ shard sets. Let $F_i(j,u)$ count the subsets of size $j$
among the first $i$ shards whose aggregate distance is $u$. Starting from
$F_0(0,0)=1$, we apply
$F_i(j,u)=F_{i-1}(j,u)+F_{i-1}(j-1,u-d_i(q))$
and obtain $w_q(u)=F_M(t,u)/\binom{M}{t}$. For identifier-fixed assignment,
$w_q$ is a point mass at the realized aggregate distance.

For every $u$, we evaluate the exact randomized-response privacy-loss
distribution in log space and obtain its hockey-stick profile
$\delta_u^{\mathrm{RR}}(\varepsilon)$. The route-averaged profile is then
formed using $w_q(u)$ and inverted by bisection to obtain the pair-specific
XDP value at target $\delta_0$. Both likelihood-ratio directions are
evaluated.

\subsubsection{Query Accountants}

For access-pattern privacy, the value in
Theorem~\ref{thm:query-transcript} is computed directly from the PRR
probability $a$, IRR probability $b$, repetition count $R$, and fixed-hash
query distance $D_Q$. Log-sum-exp arithmetic is used when evaluating
$\eta_R$, jointly accounting for the memoized PRR state and all fresh IRR
reports.

For search-pattern privacy, the reuse and split distributions depend only on
the number of ones among the first $R-1$ reports and the final report bit. We
enumerate these possibilities to compute $\gamma_{=}$ and $\gamma_{\ne}$
exactly. No additional access- or search-trace experiment is required because,
under the stated system model, the HNSW trace is post-processing of the
perturbed reports and fixed index. This conclusion excludes stable query
identifiers, routing based on clean queries, and external side channels.

\subsection{Empirical Cross-Shard Inference Evaluation}
\label{app:empirical-leakage-evaluation}

\subsubsection{Evaluation Objective}

The formal complete-view theorem grants the server perfect association of all
shard-local entries belonging to the target item. This assumption is
conservative because the implementation exposes no direct cross-shard
correspondence: different shards use different ISO-LSH mappings and local
handles, randomized response is sampled independently, and payload
ciphertexts use fresh encryption randomness.

The experiment asks a more practical question: starting from one known target
entry, how many additional target entries can a specified inference procedure
identify, and how much XDP evidence can it organize after combining the
recovered observations? The experiment does not replace the mechanism-level
XDP theorem. It measures only the leakage achieved by the evaluated
procedures, challenge pairs, auxiliary seeds, and system configuration.

\subsubsection{Neighboring Worlds}

For each challenge pair $q=(\bm x_0,\bm x_1)$, we construct two neighboring
databases with identical background records and identifiers. They differ only
in the embedding of one target identifier $i^\star$: world~0 contains
$(i^\star,\bm x_0)$ and world~1 contains
$(i^\star,\bm x_1)$.
Every trial reruns $t$-of-$M$ assignment, storage-side randomized response,
shard-local graph construction, and cross-shard inference. The assignment
mechanism is identical in the two worlds and independent of the protected
embedding. Calibration and held-out trials use disjoint randomness.
Ground-truth cross-shard correspondences and world labels are never supplied
to the inference procedures.

\subsubsection{Cross-Shard Inference Procedures}

We evaluate four complementary procedures. MDS constructs shard-local
Hamming-distance matrices and obtains low-dimensional coordinates using
classical multidimensional scaling~\cite{torgerson1952multidimensional};
the coordinates are then aligned using fixed background seeds. Topology
compares node degree, neighbor-degree statistics, and multi-hop HNSW
structure, drawing on topology-based de-anonymization and seeded graph
matching~\cite{narayanan2009deanonymizing,fishkind2019seeded,
kazemi2015growing}. GSGM propagates perturbed-code signals over each graph
and compares the resulting graph-filtered
descriptors~\cite{liu2024blind}. Joint standardizes and combines the three
scores and applies cycle-consistency checks across shard
triples~\cite{swoboda2019convex}.

The evaluator provides one known target entry in a reference shard and a fixed
fraction of background correspondences as auxiliary seeds. The remaining
target entries are hidden. For each destination shard, a procedure ranks all
non-seed candidates using only server-visible perturbed codes and graph
structure. It then retains the highest-confidence candidates until the public
selected-shard count $t$ is reached. Incorrect candidates are not removed
using ground truth, and no learned classifier is used. All descriptors,
normalizations, weights, candidate-count grids, and stopping rules are fixed
before held-out evaluation.

\subsubsection{From Cross-Shard Association to XDP}

The central step is to convert the inferred cross-shard
associations into a fixed-length multi-shard observation. For each trial, the
$t$ selected $\kappa$-bit blocks are concatenated into a reconstructed vector
of length $t\kappa$.

If an inferred entry is the correct target entry in shard $s$, its block
contains genuine randomized-response evidence about the challenge pair and
contributes the actual fixed-hash distance $d_s(q)$. If the inferred entry is
incorrect, the selected background block remains in the reconstruction. Such
a block is expected to behave like an unrelated random code, whose normalized
Hamming distance is approximately $1/2$. It therefore contributes no direct
target-dependent RR distance under the common-null model.

Ground truth is used only after inference to define
$
C_{\mathcal A,s}=1
$
when method $\mathcal A$ correctly identifies the target in shard $s$, and
zero otherwise. The direct target distance organized by the method is
$
U_{\mathcal A,q}
=
\sum_s C_{\mathcal A,s}d_s(q).
$
Thus, a successful association contributes its actual shard-specific
distance, while a false or missing association contributes no target
distance. False blocks are nevertheless retained when checking whether their
empirical normalized distance is close to the random-code reference $1/2$.

Across held-out trials, let
$\widehat w_{\mathcal A,q}(u)$ be the empirical distribution of
$U_{\mathcal A,q}$. We combine this distribution with the exact RR profile:
\[
\begin{aligned}
\delta_{\mathcal A,q}^{\mathrm{link}}(\varepsilon)
&=
\sum_u
\widehat w_{\mathcal A,q}(u)
\delta_u^{\mathrm{RR}}(\varepsilon),\\
\xi_{\mathcal A,q}^{\mathrm{link}}(\delta_0)
&=
\inf\left\{
\varepsilon:
\delta_{\mathcal A,q}^{\mathrm{link}}(\varepsilon)
\leq\delta_0
\right\}.
\end{aligned}
\]

This construction directly reflects the experimental objective. If only the
supplied reference entry is retained, the result contains approximately
single-shard XDP evidence. Each additional correct association adds its actual
fixed-hash distance and increases the reconstructed multi-shard XDP. If all
target entries are identified, the reconstructed view approaches the
idealized complete-view observation.
Importantly, the association success rate is not treated as a subsampling
probability and is not multiplied directly by a privacy budget. Instead, the
experiment first constructs the empirical distribution of the total recovered
distance and then applies the exact RR privacy-loss accountant. Consequently,
$\xi_{\mathcal A,q}^{\mathrm{link}}$ remains a distance-aware XDP quantity.

Because the candidate-selection procedure is data dependent,
$\xi^{\mathrm{link}}$ should be interpreted as the direct RR evidence organized
by the specified inference method, rather than a mechanism-level guarantee
against arbitrary algorithms.

\subsubsection{Configuration}

The experiments use $M=64$ shards, $t=16$ selected shards per item,
$\kappa=128$ bits per shard, and storage-side flip probability $p_D=0.08$.
At angular radius $d_\theta\leq0.05$, each valid configuration contains up to
five real challenge pairs, one known target entry, and $20\%$ auxiliary
background seeds. Every challenge pair uses 30 calibration trials and 300
held-out trials per world.

We evaluate index prefixes containing 5,000, 10,000, and 20,000 vectors.
MS MARCO and SIFT contain five eligible pairs in each valid configuration,
whereas TripClick at 10k contains two. The prior random-LSH regression value
$\xi=20.041$ for $p=0.21$, $\kappa=128$, $d_\theta=0.05$, and
$\delta_0=0.01$ is used only as an implementation regression test. It is not
substituted for the fixed-ISO-LSH accountant used in the evaluation.

\subsubsection{Results and Analysis}

At $d_\theta\leq0.05$ and $\delta_0=0.01$, the largest reconstructed XDP value
over MDS, Topology, GSGM, and Joint ranges from $31.586$ to $78.110$ across
the valid configurations. These values are above the corresponding
single-shard diagnostic range of $24.423$--$36.635$, showing that the evaluated
procedures can sometimes organize evidence from more than the supplied
reference shard. However, they remain far below the idealized complete-$16$-
shard privacy-loss thresholds of $200.272$--$410.314$. The practical inference
procedures therefore recover only part of the information granted to the
perfect-association adversary in the formal theorem.

For MS MARCO, the largest reconstructed value is $71.812$ at 5k using MDS,
$60.444$ at 10k using MDS, and $70.588$ at 20k using GSGM. These results show
that residual geometry is effective at smaller scales, while graph-signal
information becomes competitive at 20k. The decrease from 5k to 10k and the
subsequent increase at 20k also show that reconstructed leakage is not
determined by database size alone. For SIFT, the maximum increases from
$60.921$ at 5k to $69.692$ at 10k and $78.110$ at 20k. GSGM produces the
largest values at 5k and 10k, whereas Joint is strongest at 20k. SIFT therefore
exhibits the clearest increase in organized cross-shard evidence as the
evaluated prefix grows, although this observation is specific to the sampled
challenge pairs and graph instances. TripClick produces a maximum of $31.586$
at 10k using Joint, close to the single-shard range and substantially below the
MS MARCO and SIFT results. Because only two qualifying challenge pairs are
available, this value should not be interpreted as a general property of the
complete TripClick dataset.

No inference procedure dominates all datasets and scales. MDS is strongest in
two valid configurations, GSGM in three, and Joint in two. Taking the maximum
over the four procedures is therefore necessary, since reporting Joint alone
would underestimate the leakage achieved in several configurations. The
normalized distances of incorrectly selected blocks are generally close to
$1/2$, supporting the common-null assumption that failed associations behave
approximately like unrelated binary codes. The reconstructed XDP is also
non-monotone in database size because the eligible challenge pairs, background
candidate pools, selected shards, and independently constructed graphs change
across configurations.

LAION does not contribute a valid value because its 5k, 10k, and 20k runs
terminated on an invalid angular-distance value. TripClick at 5k contains no
eligible real pair, while its 20k run terminated on the same validation error.
These configurations are excluded rather than assigned extrapolated privacy
values.

Overall, the experiment supports two conclusions. First, cross-shard inference
can increase the XDP evidence available beyond a single known shard. Second,
the evaluated procedures remain substantially weaker than the perfect
association assumed by the complete-view theorem. The reported
$\xi^{\mathrm{link}}$ values therefore quantify the practical leakage achieved
by the specified inference procedures, while the mechanism-level theorem
remains the formal security statement.

\subsection{Experimental Parameters}
\label{app:experimental-parameters}

\noindent
For a fair ablation comparison, the No-Noise baseline and \name~use the same
multi-graph construction and HNSW parameters. Unless otherwise specified, we set
the number of graph shards to \(M=64\), and route each database point to
\(t=16\) selected shards. The HNSW construction parameter is fixed to
\(\texttt{efConstruction}=300\) across all four datasets. Compass is evaluated
using the configuration reported in the original Compass paper.

For candidate generation, No-Noise uses smaller candidate pools because it
searches over unperturbed binary embeddings. Its candidate pool ranges are
\(250\)--\(400\) on SIFT, \(200\)--\(300\) on LAION,
\(900\)--\(1050\) on TripClick, and \(1400\)--\(1600\) on MS MARCO. In contrast,
\name~uses larger candidate pools to compensate for query-side and storage-side
perturbation: \(1400\)--\(1600\) on SIFT, \(700\)--\(800\) on LAION,
\(2900\)--\(3200\) on TripClick, and \(6000\)--\(6500\) on MS MARCO.

\begin{figure*}[!t]
\centering
\noindent\fbox{%
\begin{minipage}{\dimexpr\textwidth-2\fboxsep-2\fboxrule\relax}
\vspace{1.5mm}

\noindent \textbf{Public Parameters:} dataset size $N$, vector dimension $d$,
number of graph shards $M$, routing multiplicity $t$, shard hash lengths
$\{\kappa_m\}_{m=1}^{M}$, storage-side flip probability $p_D$, query-side PRR
flip probability $p_{\mathrm{PRR}}$, query-side IRR flip probability
$p_{\mathrm{IRR}}$, per-shard candidate budget $P$, final result size $K$,
and HNSW parameters.\\
\textbf{Entities:} Client $\mathcal C$ and cloud server $\mathcal S$ in the
semi-honest setting.\\
\textbf{Input:} At setup, $\mathcal C$ holds a plaintext database
$\mathcal D=\{(\bm x_i,\mathrm{ID}_i)\}_{i=1}^{N}$, a fixed mapping-training
set $\mathcal D_{\mathrm{train}}$ that is public or disjoint from
$\mathcal D$, shard seeds $\{\sigma_m\}_{m=1}^{M}$, a routing seed
$\sigma_{\mathrm{route}}$, and shard encryption keys
$\{sk_m\}_{m=1}^{M}$. At query time, $\mathcal C$ holds a query embedding
$\bm q$ and a local memoization cache $\mathcal C_{\mathrm{memo}}$.\\
\textbf{Server State:} After setup, $\mathcal S$ stores the outsourced
multi-graph index $\mathbb G=\{\mathcal G^{(m)}\}_{m=1}^{M}$ and the attached
encrypted payloads.\\
\textbf{Output:} $\mathcal C$ outputs the identifiers and exact $\ell_2$
distances of the Top-$K$ reranked candidates.

\vspace{1.5mm}
\hrule
\vspace{1.5mm}

\noindent \textbf{[Initialization]}
\begin{enumerate}[label=\arabic*:, leftmargin=0.75cm, itemsep=0.8mm, topsep=1mm]
    \item For each shard $m$, $\mathcal C$ invokes $\mathsf{IsoHash}$ on
    $\mathcal D_{\mathrm{train}}$ using shard seed $\sigma_m$ to derive the
    shard-specific hash state
    $\mathsf{st}_{\mathrm{hash}}^{(m)}$. It obtains
    $\mathsf{st}_{\mathrm{hash}}
    =\{\mathsf{st}_{\mathrm{hash}}^{(m)}\}_{m=1}^{M}$.

    \item For each database point $(\bm x_i,\mathrm{ID}_i)$, $\mathcal C$
    obtains its routed shard set
    $\mathcal S_i\leftarrow
    \mathsf{Route}(\mathrm{ID}_i,M,t;\sigma_{\mathrm{route}})$.
    For each shard $m\in\mathcal S_i$, $\mathcal C$ computes
    $\bm v_{i,m}\leftarrow
    \mathsf{Hash}(\bm x_i;\mathsf{st}_{\mathrm{hash}}^{(m)})$,
    perturbs it as
    $\widetilde{\bm v}_{i,m}\leftarrow
    \mathsf{LSHRR}(\bm v_{i,m},p_D)$,
    assigns a fresh unique shard-local label $\ell_{i,m}$, samples fresh
    encryption randomness $\rho_{i,m}$, and computes
    $C_{i,m}\leftarrow
    \mathsf{Enc}_{sk_m}(\mathrm{ID}_i\Vert\bm x_i;\rho_{i,m})$.
    It uploads
    $(\ell_{i,m},\widetilde{\bm v}_{i,m},C_{i,m})$ to $\mathcal S$
    and stores $\{(m,\ell_{i,m})\}_{m\in\mathcal S_i}$ locally for deletion.

    \item For each shard $m$, $\mathcal S$ constructs an HNSW graph shard
    $\mathcal G^{(m)}$ over the received perturbed binary embeddings using
    Hamming distance. Each node is indexed by its shard-local label and
    stores the corresponding encrypted payload.
\end{enumerate}

\vspace{1.5mm}
\hrule
\vspace{1.5mm}

\noindent \textbf{[Query Generation]}
\begin{enumerate}[label=\arabic*:, leftmargin=0.75cm, itemsep=0.8mm,
topsep=1mm, start=4]
    \item Given a query embedding $\bm q$, $\mathcal C$ computes a
    shard-specific query binary embedding for each shard:
    $\bm v_q^{(m)}\leftarrow
    \mathsf{Hash}(\bm q;\mathsf{st}_{\mathrm{hash}}^{(m)})$
    for $m\in[M]$.

    \item For each shard $m$, $\mathcal C$ checks the local memoization cache.
    If a cached PRR representation exists, it retrieves
    $\overline{\bm v}_q^{(m)}
    \leftarrow\mathcal C_{\mathrm{memo}}[\bm q,m]$.
    Otherwise, it samples
    $\overline{\bm v}_q^{(m)}
    \leftarrow\mathsf{LSHRR}
    (\bm v_q^{(m)},p_{\mathrm{PRR}})$
    and stores it in $\mathcal C_{\mathrm{memo}}$.

    \item Before transmission, $\mathcal C$ applies instantaneous
    perturbation to each memoized representation:
    $\widetilde{\bm v}_q^{(m)}
    \leftarrow\mathsf{LSHRR}
    (\overline{\bm v}_q^{(m)},p_{\mathrm{IRR}})$.
    It sends the perturbed query binary embeddings
    $\{\widetilde{\bm v}_q^{(m)}\}_{m=1}^{M}$ to $\mathcal S$ as the
    retrieval request.
\end{enumerate}

\vspace{1.5mm}
\hrule
\vspace{1.5mm}

\noindent \textbf{[Server-Side Search]}
\begin{enumerate}[label=\arabic*:, leftmargin=0.75cm, itemsep=0.8mm,
topsep=1mm, start=7]
    \item For each shard $m$, $\mathcal S$ performs HNSW search over
    $\mathcal G^{(m)}$ using $\widetilde{\bm v}_q^{(m)}$ as the query binary
    embedding and obtains a shard-level candidate-label set
    $\mathrm{Cand}^{(m)}$ with budget $P$.

    \item $\mathcal S$ collects the encrypted payloads associated with the
    returned labels and sends the aggregated response
    $\mathrm{Cand}
    =\bigcup_{m=1}^{M}
    \{(m,C_{i,m}):\ell_{i,m}\in\mathrm{Cand}^{(m)}\}$
    to $\mathcal C$. Each ciphertext is accompanied by its shard index for
    selecting the corresponding decryption key.
\end{enumerate}

\vspace{1.5mm}
\hrule
\vspace{1.5mm}

\noindent \textbf{[Local Decryption and Exact Reranking]}
\begin{enumerate}[label=\arabic*:, leftmargin=0.75cm, itemsep=0.8mm,
topsep=1mm, start=9]
    \item For each pair $(m,C_{i,m})$ in the returned response, $\mathcal C$
    decrypts the payload using the corresponding shard key:
    $(\mathrm{ID}_i,\bm x_i)
    \leftarrow\mathsf{Dec}_{sk_m}(C_{i,m})$.

    \item $\mathcal C$ removes duplicate candidates caused by multi-graph
    routing, computes the exact Euclidean distance between $\bm q$ and each
    recovered embedding $\bm x_i$, and reranks the candidates according to
    their exact distances.

    \item Finally, $\mathcal C$ outputs the Top-$K$ reranked candidates
    together with their identifiers and exact distances.
\end{enumerate}

\vspace{1mm}
\end{minipage}%
}
\caption{The complete protocol of \name, including initialization, PRR/IRR
query generation, server-side HNSW search, and client-side decryption and
reranking.}
\label{protocolflow}
\end{figure*}

\section{Correctness of the MESS}
\subsection{Distribution of the Shard-Level Rank}
\label{app:rank-derivation}

\noindent
In this appendix, we derive the mean, variance, and asymptotic normal
approximation of the shard-level rank used in the correctness analysis. We
also explain its connection to the empirical Recall and MRR models.

\begin{proposition}[Poisson Binomial Model of Shard-Level Rank]
\label{prop:appendix-shard-rank}
\normalfont
Consider one queried shard \(s\). Let \(Rank_s\) denote the absolute rank of the target
item after storage-side and query-side perturbation. Let \(N(d)\) be the number
of non-target database items whose clean Hamming distance to the query is
\(d\), where \(d\in\{0,\ldots,\kappa\}\). Since every item is routed uniformly
to \(t\) of the \(M\) shards, the expected number of distance-\(d\) items in
one shard is
\(\mathbb E[N_s(d)]=(t/M)N(d)\).

For a background item at distance \(d\), let
\(I_{i,d}=\mathbf 1\{\widetilde D_{i,d}<\widetilde D_{d_s^*}\}\) be its
overtaking indicator, where \(d_s^*\) is the clean distance of the target.
The effective query flip probability is
\(p_q=p_{\mathrm{PRR}}(1-p_{\mathrm{IRR}})
+p_{\mathrm{IRR}}(1-p_{\mathrm{PRR}})\), and the compound flip probability is
\(p_{\oplus}=p_D(1-p_q)+p_q(1-p_D)\).
The Gaussian approximation gives
\[
P_{\mathrm{overtake}}(d)
\approx
\Phi\!\left(
\frac{(d_s^*-d)(1-2p_{\oplus})}
{\sqrt{2\kappa p_{\oplus}(1-p_{\oplus})}}
\right),
\]
consistent with Proposition~\ref{prop:candidate-correctness}.

Under the expected-occupancy and independent-overtaking approximations, the
overtaking count is Poisson binomial with
\[
\begin{aligned}
\mu_{\mathrm{shard}}
&=
\frac{t}{M}\sum_{d=0}^{\kappa}
N(d)P_{\mathrm{overtake}}(d),\\
\sigma_{\mathrm{shard}}^2
&=
\frac{t}{M}\sum_{d=0}^{\kappa}
N(d)P_{\mathrm{overtake}}(d)
\bigl(1-P_{\mathrm{overtake}}(d)\bigr).
\end{aligned}
\]
When sufficiently many background items contribute and no individual term
dominates the variance,
\(Rank_s\approx
\mathcal N(1+\mu_{\mathrm{shard}},\sigma_{\mathrm{shard}}^2)\).
\end{proposition}

\noindent\textit{Proof.}
Since every item is routed uniformly to \(t\) of the \(M\) shards, it appears
in any fixed shard with probability \(t/M\). Thus, a distance layer containing
\(N(d)\) background items contributes an expected
\((t/M)N(d)\) items to that shard.

For a background item at clean distance \(d\), the storage-side and query-side
perturbations form an effective binary symmetric channel with crossover
probability \(p_{\oplus}\). Among the \(d\) initially different coordinates,
the remaining mismatch count follows
\(X_d\sim\operatorname{Binomial}(d,1-p_{\oplus})\). Among the other
\(\kappa-d\) coordinates, newly created mismatches follow
\(Y_d\sim\operatorname{Binomial}(\kappa-d,p_{\oplus})\).
Therefore, \(\widetilde D_d=X_d+Y_d\), with mean
\(\kappa p_{\oplus}+d(1-2p_{\oplus})\) and variance
\(\kappa p_{\oplus}(1-p_{\oplus})\).

Under the standard independence approximation, the difference between the
perturbed distances of the background item and the target has mean
\((d-d_s^*)(1-2p_{\oplus})\) and variance
\(2\kappa p_{\oplus}(1-p_{\oplus})\). Applying the Gaussian approximation
gives \(P_{\mathrm{overtake}}(d)\). Hence,
\(I_{i,d}\sim
\operatorname{Bernoulli}(P_{\mathrm{overtake}}(d))\), and the target rank is
\(Rank_s=1+\sum_d\sum_i I_{i,d}\).

Because the Bernoulli parameters depend on \(d\), the overtaking count is
Poisson binomial. Linearity of expectation gives
\(\mathbb E[Rank_s]=1+(t/M)\sum_dN(d)P_{\mathrm{overtake}}(d)\).
Under the independence approximation, its variance is
\((t/M)\sum_dN(d)P_{\mathrm{overtake}}(d)
(1-P_{\mathrm{overtake}}(d))\), giving
\(\mu_{\mathrm{shard}}\) and \(\sigma_{\mathrm{shard}}^2\).
The central limit approximation for Poisson binomial variables then yields
\(Rank_s\approx
\mathcal N(1+\mu_{\mathrm{shard}},\sigma_{\mathrm{shard}}^2)\).
\hfill\(\square\)

The factor \(t/M\) above models expected shard occupancy for correctness. It
is not treated as a subsampling probability in the complete-server-view
privacy analysis.

\smallskip\noindent
\textbf{Candidate Recovery.}
For a per-shard candidate budget \(P\), the target-inclusion probability is
approximately
\(P_{\mathrm{hit},s}(P)=
\Phi((P+\frac12-(1+\mu_{\mathrm{shard}}))
/\sigma_{\mathrm{shard}})\), where \(1/2\) is the continuity correction.
If the experiment specifies a total candidate pool \(B\) distributed evenly
across \(M\) shards, the per-shard budget is \(\lfloor B/M\rfloor\); we write
\(P_{\mathrm{hit},s}(B)=
P_{\mathrm{hit},s}(\lfloor B/M\rfloor)\).

For an item stored in the shard set \(\mathcal S_i\), conditional independence
across shards gives
\(P_{\mathrm{hit},i}(B)\approx
1-\prod_{s\in\mathcal S_i}(1-P_{\mathrm{hit},s}(B))\).
Thus, the item is missed only when it is not recovered from any shard
containing it.

The implementation returns candidates in multiples of \(M\) before duplicate
removal. Therefore, the realized aggregate pool corresponding to configured
budget \(B\) is \(B_{\mathrm{eff}}(B)=M\lfloor B/M\rfloor+1/2\), including
the continuity correction. The validation plots use this aggregate threshold
to fit an effective normal-CDF model to the observed end-to-end metric. The
parameters \(\mu_R\) and \(\sigma_R\) reported in the plots are fitted
aggregate-rank parameters; they need not equal the analytical shard-level
moments above.

\smallskip\noindent
\textbf{Recall Model.}
Absolute rank alone does not fully determine empirical Recall. The target must
also be reached by HNSW traversal, returned in a shard-level candidate set,
survive aggregation and duplicate removal, and enter exact reranking.
Therefore, empirical Recall may saturate below one.

We introduce a coverage coefficient
\(\gamma_{\mathrm{cov}}\in[0,1]\) to capture these implementation-level
effects. The fitted Recall curve is
\[
R(B)
=
\gamma_{\mathrm{cov}}
\Phi\!\left(
\frac{B_{\mathrm{eff}}(B)-\mu_R}{\sigma_R}
\right).
\]
Here, \(\gamma_{\mathrm{cov}}\) captures implementation-level candidate
coverage and the empirical saturation level. It is not a parameter of the
analytical shard-level rank distribution. The Recall panels of Figure~\ref{fig:rank-metric-fits} compare this model
with the measurements.

\smallskip\noindent
\textbf{MRR Model.}
For MS MARCO and TripClick, MRR@10 depends not only on whether a relevant
document is recovered, but also on its final position after exact reranking.
For a query \(q\), let \(r_q\) be the rank of its first relevant document.
Its reciprocal-rank contribution is \(1/r_q\) when \(r_q\leq10\) and zero
otherwise. The dataset-level metric is
\(\mathrm{MRR@10}=|\mathcal Q|^{-1}
\sum_{q\in\mathcal Q}\mathrm{RR@10}(q)\).

Let \(H_q(B)\) be the probability that at least one relevant document is
recovered, and let \(c_q(B)\) denote the conditional expectation of
\((1/r_q)\mathbf 1\{r_q\leq10\}\) given this recovery event. Then
\[
\mathbb E[\mathrm{RR@10}(q)]
=
H_q(B)c_q(B).
\]
Approximating the conditional ranking contribution by an empirical coefficient
\(\gamma_{\mathrm{mrr}}\in[0,1]\) gives the fitted curve
\[
\mathrm{MRR@10}(B)
=
\gamma_{\mathrm{mrr}}
\Phi\!\left(
\frac{B_{\mathrm{eff}}(B)-\mu_R}{\sigma_R}
\right).
\]
The coefficient \(\gamma_{\mathrm{mrr}}\) captures empirical effects from
candidate coverage, aggregation, duplicate removal, exact reranking, and the
Top-10 evaluation window. Like \(\gamma_{\mathrm{cov}}\), it is not a
parameter of the rank distribution or a formal correctness guarantee.
The MRR panels of Figure~\ref{fig:rank-metric-fits} compare this model
with the measurements. For every dataset and perturbation configuration,
\(\gamma\), \(\mu_R\), and \(\sigma_R\) are jointly estimated from all
measured points by nonlinear least squares.

The derivation relies on Gaussian distance approximation, independent
overtaking events, and conditional independence across selected shards. It
therefore models the observed retrieval trends rather than exactly
characterizing HNSW execution.

\begin{figure*}[p]
\centering
\captionsetup[subfigure]{font=scriptsize,skip=2pt}
\definecolor{ObsBlue}{RGB}{64,105,166}
\definecolor{FitRed}{RGB}{196,78,82}

\begin{subfigure}[t]{0.45\textwidth}
\centering
\pgfmathsetmacro{\gammaSiftA}{0.908057}
\pgfmathsetmacro{\muSiftA}{113.488662}
\pgfmathsetmacro{\sigmaSiftA}{499.221899}

\begin{tikzpicture}
\begin{axis}[
    width=\linewidth,
    height=0.5625\linewidth,
    xlabel={Candidate Pool},
    ylabel={Recall},
    xmin=0,
    xmax=2600,
    ymin=0.28,
    ymax=0.96,
    grid=major,
    grid style={gray!25, line width=0.2pt},
    axis lines=box,
    tick label style={font=\scriptsize},
    label style={font=\scriptsize},
    legend to name=rankfitlegend,
    legend columns=2,
    legend style={
        font=\scriptsize,
        draw=none,
        column sep=8pt
    }
]
\addplot[
    only marks,
    mark=*,
    mark size=1.15pt,
    color=ObsBlue
]
coordinates {
(50,0.32100000)
    (100,0.32100000)
    (150,0.45500000)
    (200,0.54400000)
    (250,0.54400000)
    (300,0.60100000)
    (350,0.63800000)
    (400,0.67600000)
    (450,0.70600000)
    (500,0.70600000)
    (550,0.73400000)
    (600,0.75700000)
    (650,0.76400000)
    (700,0.76400000)
    (750,0.78700000)
    (800,0.80600000)
    (850,0.81900000)
    (900,0.83100000)
    (950,0.83100000)
    (1000,0.84300000)
    (1050,0.84900000)
    (1100,0.85900000)
    (1150,0.85900000)
    (1200,0.86500000)
    (1250,0.87400000)
    (1300,0.87800000)
    (1350,0.88700000)
    (1400,0.88700000)
    (1450,0.89400000)
    (1500,0.89600000)
    (1550,0.89700000)
    (1600,0.90000000)
    (1650,0.90000000)
    (1700,0.90300000)
    (1750,0.90500000)
    (1800,0.90700000)
    (1850,0.90700000)
    (1900,0.91300000)
    (1950,0.91800000)
    (2000,0.91800000)
    (2050,0.92200000)
    (2100,0.92200000)
    (2150,0.92400000)
    (2200,0.92900000)
    (2250,0.93400000)
    (2300,0.93400000)
    (2350,0.93400000)
    (2400,0.93500000)
    (2450,0.93800000)
    (2500,0.93900000)
};
\addlegendentry{Measured}

\addplot[
    color=FitRed,
    line width=0.9pt,
    domain=0:2500,
    samples=350
]
{\gammaSiftA*0.5*
 (1+erf((x-\muSiftA)/(\sigmaSiftA*sqrt(2))))};
\addlegendentry{Normal-CDF fit}

\end{axis}
\end{tikzpicture}
\caption{SIFT, $(p_D,p_q)=(0.08,0.08)$.}
\label{fig:rank-fit-sift-008-008}
\end{subfigure}
\hfill
\begin{subfigure}[t]{0.45\textwidth}
\centering
\pgfmathsetmacro{\gammaSiftB}{0.893925}
\pgfmathsetmacro{\muSiftB}{233.578403}
\pgfmathsetmacro{\sigmaSiftB}{745.936176}

\begin{tikzpicture}
\begin{axis}[
    width=\linewidth,
    height=0.5625\linewidth,
    xlabel={Candidate Pool},
    ylabel={Recall},
    xmin=0,
    xmax=4100,
    ymin=0.20,
    ymax=0.96,
    grid=major,
    grid style={gray!25, line width=0.2pt},
    axis lines=box,
    tick label style={font=\scriptsize},
    label style={font=\scriptsize}
]
\addplot[
    only marks,
    mark=*,
    mark size=1.15pt,
    color=ObsBlue
]
coordinates {
(50,0.24800000)
    (100,0.24800000)
    (150,0.35800000)
    (200,0.43200000)
    (250,0.43200000)
    (300,0.47800000)
    (350,0.52400000)
    (400,0.55500000)
    (450,0.58700000)
    (500,0.58700000)
    (550,0.61600000)
    (600,0.64000000)
    (650,0.65800000)
    (700,0.65800000)
    (750,0.68000000)
    (800,0.70200000)
    (850,0.72100000)
    (900,0.73900000)
    (950,0.73900000)
    (1000,0.75700000)
    (1050,0.76600000)
    (1100,0.77200000)
    (1150,0.77200000)
    (1200,0.78300000)
    (1250,0.79500000)
    (1300,0.80300000)
    (1350,0.81000000)
    (1400,0.81000000)
    (1450,0.81700000)
    (1500,0.81800000)
    (1550,0.82100000)
    (1600,0.82700000)
    (1650,0.82700000)
    (1700,0.83400000)
    (1750,0.83600000)
    (1800,0.83800000)
    (1850,0.83800000)
    (1900,0.84700000)
    (1950,0.85400000)
    (2000,0.85900000)
    (2050,0.85900000)
    (2100,0.85900000)
    (2150,0.86600000)
    (2200,0.87000000)
    (2250,0.87000000)
    (2300,0.87000000)
    (2350,0.87300000)
    (2400,0.87600000)
    (2450,0.88000000)
    (2500,0.88600000)
    (2550,0.88600000)
    (2600,0.89000000)
    (2650,0.89200000)
    (2700,0.89800000)
    (2750,0.89800000)
    (2800,0.90100000)
    (2850,0.90300000)
    (2900,0.90400000)
    (2950,0.90600000)
    (3000,0.90600000)
    (3050,0.90700000)
    (3100,0.90800000)
    (3150,0.91000000)
    (3200,0.91000000)
    (3250,0.91000000)
    (3300,0.91000000)
    (3350,0.91400000)
    (3400,0.91500000)
    (3450,0.91500000)
    (3500,0.91900000)
    (3550,0.92200000)
    (3600,0.92200000)
    (3650,0.92100000)
    (3700,0.92100000)
    (3750,0.92300000)
    (3800,0.92500000)
    (3850,0.92500000)
    (3900,0.92500000)
    (3950,0.92900000)
    (4000,0.93200000)
};

\addplot[
    color=FitRed,
    line width=0.9pt,
    domain=0:4000,
    samples=450
]
{\gammaSiftB*0.5*
 (1+erf((x-\muSiftB)/(\sigmaSiftB*sqrt(2))))};

\end{axis}
\end{tikzpicture}
\caption{SIFT, $(p_D,p_q)=(0.12,0.08)$.}
\label{fig:rank-fit-sift-012-008}
\end{subfigure}

\par\vspace{-1mm}

\begin{subfigure}[t]{0.45\textwidth}
\centering
\pgfmathsetmacro{\gammaLaionA}{0.916696}
\pgfmathsetmacro{\muLaionA}{13.256833}
\pgfmathsetmacro{\sigmaLaionA}{306.073697}

\begin{tikzpicture}
\begin{axis}[
    width=\linewidth,
    height=0.5625\linewidth,
    xlabel={Candidate Pool},
    ylabel={Recall},
    xmin=0,
    xmax=2100,
    ymin=0.40,
    ymax=0.96,
    grid=major,
    grid style={gray!25, line width=0.2pt},
    axis lines=box,
    tick label style={font=\scriptsize},
    label style={font=\scriptsize}
]
\addplot[
    only marks,
    mark=*,
    mark size=1.15pt,
    color=ObsBlue
]
coordinates {
(50,0.43500000)
    (100,0.43500000)
    (150,0.62400000)
    (200,0.70200000)
    (250,0.70200000)
    (300,0.74800000)
    (350,0.78700000)
    (400,0.81100000)
    (450,0.82900000)
    (500,0.82900000)
    (550,0.84500000)
    (600,0.85900000)
    (650,0.87100000)
    (700,0.87100000)
    (750,0.88100000)
    (800,0.88800000)
    (850,0.89500000)
    (900,0.89700000)
    (950,0.89700000)
    (1000,0.90400000)
    (1050,0.91100000)
    (1100,0.91600000)
    (1150,0.91600000)
    (1200,0.91900000)
    (1250,0.92100000)
    (1300,0.92100000)
    (1350,0.92300000)
    (1400,0.92300000)
    (1450,0.92300000)
    (1500,0.92400000)
    (1550,0.92500000)
    (1600,0.92700000)
    (1650,0.92700000)
    (1700,0.93000000)
    (1750,0.93300000)
    (1800,0.93400000)
    (1850,0.93400000)
    (1900,0.93600000)
    (1950,0.93800000)
    (2000,0.93900000)
};

\addplot[
    color=FitRed,
    line width=0.9pt,
    domain=0:2000,
    samples=350
]
{\gammaLaionA*0.5*
 (1+erf((x-\muLaionA)/(\sigmaLaionA*sqrt(2))))};

\end{axis}
\end{tikzpicture}
\caption{LAION, $(p_D,p_q)=(0.08,0.08)$.}
\label{fig:rank-fit-laion-008-008}
\end{subfigure}
\hfill
\begin{subfigure}[t]{0.45\textwidth}
\centering
\pgfmathsetmacro{\gammaLaionB}{0.924645}
\pgfmathsetmacro{\muLaionB}{-5.027759}
\pgfmathsetmacro{\sigmaLaionB}{445.006330}

\begin{tikzpicture}
\begin{axis}[
    width=\linewidth,
    height=0.5625\linewidth,
    xlabel={Candidate Pool},
    ylabel={Recall},
    xmin=0,
    xmax=3100,
    ymin=0.38,
    ymax=0.98,
    grid=major,
    grid style={gray!25, line width=0.2pt},
    axis lines=box,
    tick label style={font=\scriptsize},
    label style={font=\scriptsize}
]
\addplot[
    only marks,
    mark=*,
    mark size=1.15pt,
    color=ObsBlue
]
coordinates {
(50,0.42000000)
    (100,0.42000000)
    (150,0.59000000)
    (200,0.66500000)
    (250,0.66500000)
    (300,0.71400000)
    (350,0.75200000)
    (400,0.77800000)
    (450,0.79500000)
    (500,0.79500000)
    (550,0.81200000)
    (600,0.82200000)
    (650,0.83500000)
    (700,0.83500000)
    (750,0.84300000)
    (800,0.85300000)
    (850,0.86300000)
    (900,0.87600000)
    (950,0.87600000)
    (1000,0.88400000)
    (1050,0.88800000)
    (1100,0.89200000)
    (1150,0.89200000)
    (1200,0.89400000)
    (1250,0.89700000)
    (1300,0.90000000)
    (1350,0.90600000)
    (1400,0.90600000)
    (1450,0.90900000)
    (1500,0.91000000)
    (1550,0.91300000)
    (1600,0.91700000)
    (1650,0.91700000)
    (1700,0.92000000)
    (1750,0.92200000)
    (1800,0.92300000)
    (1850,0.92300000)
    (1900,0.92400000)
    (1950,0.92800000)
    (2000,0.93200000)
    (2050,0.93300000)
    (2100,0.93300000)
    (2150,0.93500000)
    (2200,0.93700000)
    (2250,0.93800000)
    (2300,0.93800000)
    (2350,0.94000000)
    (2400,0.94000000)
    (2450,0.94000000)
    (2500,0.94200000)
    (2550,0.94200000)
    (2600,0.94400000)
    (2650,0.94400000)
    (2700,0.94800000)
    (2750,0.94800000)
    (2800,0.94800000)
    (2850,0.95200000)
    (2900,0.95300000)
    (2950,0.95400000)
    (3000,0.95400000)
};

\addplot[
    color=FitRed,
    line width=0.9pt,
    domain=0:3000,
    samples=400
]
{\gammaLaionB*0.5*
 (1+erf((x-\muLaionB)/(\sigmaLaionB*sqrt(2))))};

\end{axis}
\end{tikzpicture}
\caption{LAION, $(p_D,p_q)=(0.12,0.08)$.}
\label{fig:rank-fit-laion-012-008}
\end{subfigure}

\par\vspace{-1mm}

\begin{subfigure}[t]{0.45\textwidth}
\centering
\pgfmathsetmacro{\gammaMarcoA}{0.363058}
\pgfmathsetmacro{\muMarcoA}{-448.038615}
\pgfmathsetmacro{\sigmaMarcoA}{3472.259966}

\begin{tikzpicture}
\begin{axis}[
    width=\linewidth,
    height=0.5625\linewidth,
    xlabel={Candidate Pool},
    ylabel={MRR@10},
    xmin=0,
    xmax=15500,
    ymin=0.18,
    ymax=0.39,
    grid=major,
    grid style={gray!25, line width=0.2pt},
    axis lines=box,
    tick label style={font=\scriptsize},
    label style={font=\scriptsize}
]
\addplot[
    only marks,
    mark=*,
    mark size=1.15pt,
    color=ObsBlue
]
coordinates {
(500,0.206314)
    (1000,0.240152)
    (2000,0.288121)
    (3000,0.312823)
    (4000,0.324926)
    (5000,0.337466)
    (6000,0.341349)
    (7000,0.348931)
    (8000,0.351484)
    (8500,0.353013)
    (9000,0.355467)
    (9500,0.360504)
    (10000,0.363586)
    (10500,0.359360)
    (11000,0.360236)
    (11500,0.363402)
    (12000,0.364476)
    (12500,0.365778)
    (13000,0.369288)
    (13500,0.370305)
    (14000,0.368200)
    (14500,0.370747)
    (15000,0.373648)
};

\addplot[
    color=FitRed,
    line width=0.9pt,
    domain=0:15000,
    samples=550
]
{\gammaMarcoA*0.5*
 (1+erf((x-\muMarcoA)/(\sigmaMarcoA*sqrt(2))))};

\end{axis}
\end{tikzpicture}
\caption{MS MARCO, $(p_D,p_q)=(0.08,0.08)$.}
\label{fig:rank-fit-msmarco-008-008}
\end{subfigure}
\hfill
\begin{subfigure}[t]{0.45\textwidth}
\centering
\pgfmathsetmacro{\gammaMarcoB}{0.360700}
\pgfmathsetmacro{\muMarcoB}{312.804425}
\pgfmathsetmacro{\sigmaMarcoB}{2406.874878}

\begin{tikzpicture}
\begin{axis}[
    width=\linewidth,
    height=0.5625\linewidth,
    xlabel={Candidate Pool},
    ylabel={MRR@10},
    xmin=0,
    xmax=15500,
    ymin=0.14,
    ymax=0.39,
    grid=major,
    grid style={gray!25, line width=0.2pt},
    axis lines=box,
    tick label style={font=\scriptsize},
    label style={font=\scriptsize}
]
\addplot[
    only marks,
    mark=*,
    mark size=1.15pt,
    color=ObsBlue
]
coordinates {
(500,0.162141)
    (1000,0.240152)
    (2000,0.288121)
    (3000,0.312823)
    (4000,0.324926)
    (5000,0.337466)
    (6000,0.341349)
    (7000,0.348931)
    (8000,0.351484)
    (8500,0.353013)
    (9000,0.355467)
    (9500,0.360504)
    (10000,0.363586)
    (10500,0.359360)
    (11000,0.360236)
    (11500,0.363402)
    (12000,0.364476)
    (12500,0.365778)
    (13000,0.369288)
    (13500,0.370305)
    (14000,0.368200)
    (14500,0.370747)
    (15000,0.373648)
};

\addplot[
    color=FitRed,
    line width=0.9pt,
    domain=0:15000,
    samples=550
]
{\gammaMarcoB*0.5*
 (1+erf((x-\muMarcoB)/(\sigmaMarcoB*sqrt(2))))};

\end{axis}
\end{tikzpicture}
\caption{MS MARCO, $(p_D,p_q)=(0.08,0.12)$.}
\label{fig:rank-fit-msmarco-008-012}
\end{subfigure}

\par\vspace{-1mm}

\begin{subfigure}[t]{0.45\textwidth}
\centering
\pgfmathsetmacro{\gammaTripA}{0.292444}
\pgfmathsetmacro{\muTripA}{-584.845263}
\pgfmathsetmacro{\sigmaTripA}{1497.588938}

\begin{tikzpicture}
\begin{axis}[
    width=\linewidth,
    height=0.5625\linewidth,
    xlabel={Candidate Pool},
    ylabel={MRR@10},
    xmin=0,
    xmax=8200,
    ymin=0.20,
    ymax=0.31,
    grid=major,
    grid style={gray!25, line width=0.2pt},
    axis lines=box,
    tick label style={font=\scriptsize},
    label style={font=\scriptsize}
]
\addplot[
    only marks,
    mark=*,
    mark size=1.15pt,
    color=ObsBlue
]
coordinates {
(500,0.216611)
    (1000,0.256473)
    (1500,0.263782)
    (2000,0.285668)
    (2500,0.277446)
    (3000,0.286121)
    (3500,0.292065)
    (4000,0.287024)
    (4500,0.290556)
    (5000,0.293597)
    (5500,0.291062)
    (6000,0.293692)
    (6500,0.294386)
    (7000,0.295844)
    (7500,0.293397)
    (8000,0.298063)
};

\addplot[
    color=FitRed,
    line width=0.9pt,
    domain=0:8000,
    samples=450
]
{\gammaTripA*0.5*
 (1+erf((x-\muTripA)/(\sigmaTripA*sqrt(2))))};

\end{axis}
\end{tikzpicture}
\caption{TripClick, $(p_D,p_q)=(0.08,0.08)$.}
\label{fig:rank-fit-tripclick-008-008}
\end{subfigure}
\hfill
\begin{subfigure}[t]{0.45\textwidth}
\centering
\pgfmathsetmacro{\gammaTripB}{0.290457}
\pgfmathsetmacro{\muTripB}{-264.952740}
\pgfmathsetmacro{\sigmaTripB}{1571.091781}

\begin{tikzpicture}
\begin{axis}[
    width=\linewidth,
    height=0.5625\linewidth,
    xlabel={Candidate Pool},
    ylabel={MRR@10},
    xmin=0,
    xmax=8200,
    ymin=0.16,
    ymax=0.31,
    grid=major,
    grid style={gray!25, line width=0.2pt},
    axis lines=box,
    tick label style={font=\scriptsize},
    label style={font=\scriptsize}
]
\addplot[
    only marks,
    mark=*,
    mark size=1.15pt,
    color=ObsBlue
]
coordinates {
(500,0.187442)
    (1000,0.239034)
    (1500,0.255096)
    (2000,0.268584)
    (2500,0.268710)
    (3000,0.279821)
    (3500,0.286313)
    (4000,0.286638)
    (4500,0.286781)
    (5000,0.288990)
    (5500,0.293226)
    (6000,0.293044)
    (6500,0.294028)
    (7000,0.291057)
    (7500,0.295190)
    (8000,0.292808)
};

\addplot[
    color=FitRed,
    line width=0.9pt,
    domain=0:8000,
    samples=450
]
{\gammaTripB*0.5*
 (1+erf((x-\muTripB)/(\sigmaTripB*sqrt(2))))};

\end{axis}
\end{tikzpicture}
\caption{TripClick, $(p_D,p_q)=(0.08,0.12)$.}
\label{fig:rank-fit-tripclick-008-012}
\end{subfigure}

\par\vspace{-1mm}
\ref{rankfitlegend}

\caption{Measured retrieval quality and fitted normal-CDF curves. Markers show
measurements and lines show nonlinear least-squares fits over all measured
points using $B_{\mathrm{eff}}(B)=64\lfloor B/64\rfloor+1/2$. Smooth
continuous envelopes are drawn for readability.}
\label{fig:rank-metric-fits}
\end{figure*}
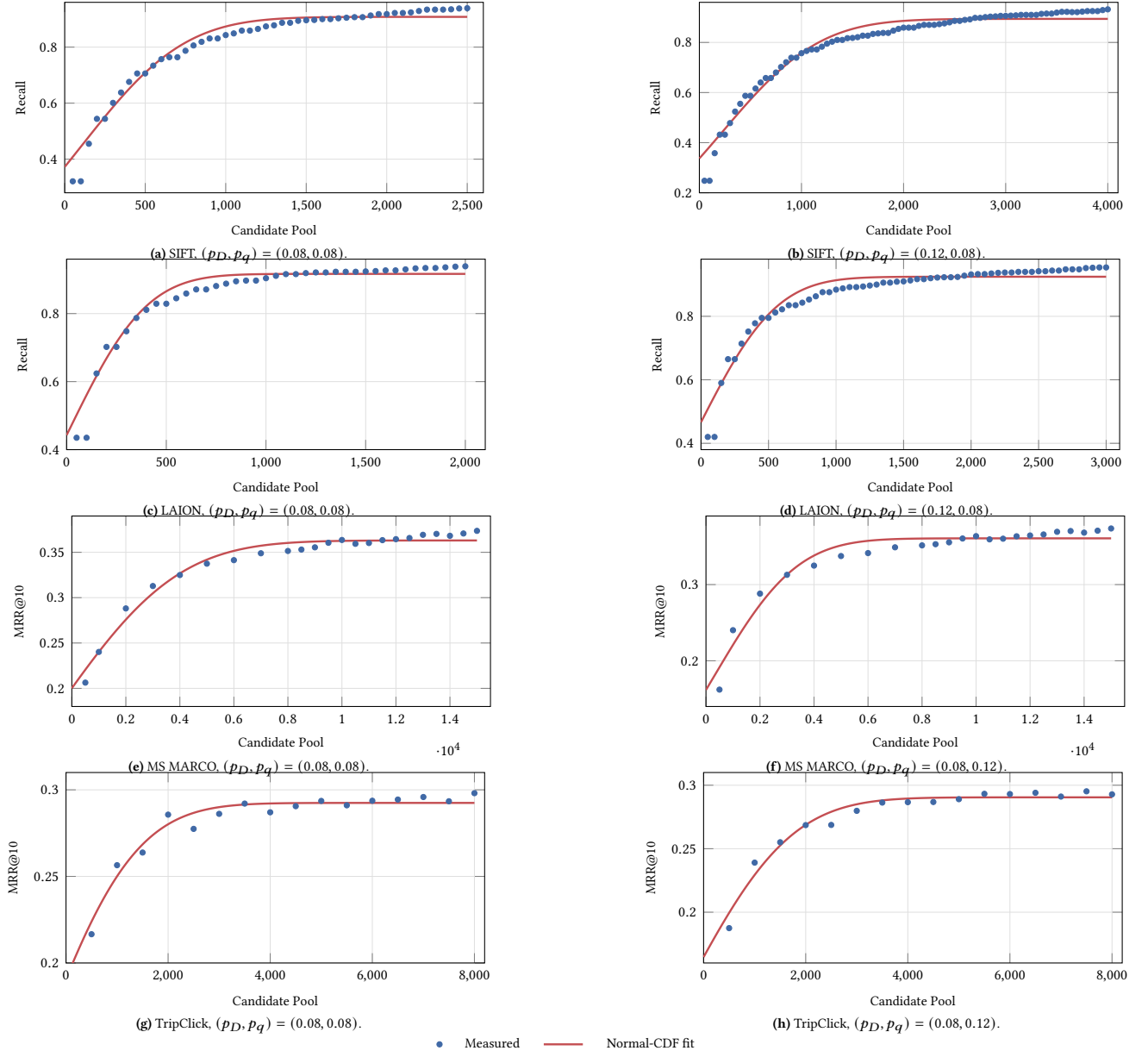

\clearpage
\begin{figure*}[p]
\centering

\begin{minipage}[t]{0.485\textwidth}
\vspace{0pt}
\begin{algorithm}[H]
\caption{MDS-Based Cross-Shard Linkage}
\label{alg:mds-linkage}
\begin{algorithmic}[1]

\Require Server-visible shards
$\mathsf V=\{\mathcal G^{(s)}=(V_s,E_s)\}_{s=1}^{M}$,
where $V_s$ and $E_s$ are the entries and HNSW edges
of shard $s$; perturbed code
$\widetilde{\bm v}_{u,s}$ of each $u\in V_s$;
fixed seed pairs
$\mathcal Z_{s_0,s}\subseteq V_{s_0}\times V_s$;
reference shard $s_0$ and target entry $a\in V_{s_0}$;
MDS dimension $r$ and threshold $\tau_{\mathrm{mds}}$

\Ensure Predicted corresponding-entry set
$\widehat{\mathcal C}_{\mathrm{mds}}(a)$ and
distinguishing score $T_{\mathrm{mds}}$

\For{$s=1,\ldots,M$}
    \State Let $D_s$ be the shard-local distance matrix,
    where
    $D_s[u,v]\gets
    d_H(\widetilde{\bm v}_{u,s},
        \widetilde{\bm v}_{v,s})$
    and $d_H$ denotes Hamming distance
    \State Let
    $X_s\gets\Call{ClassicalMDS}{D_s,r}$
    be the resulting $r$-dimensional coordinates
\EndFor

\State Initialize
$\widehat{\mathcal C}_{\mathrm{mds}}(a)
\gets\{(s_0,a)\}$
\State Let $\mathcal L$ store the accepted matching
scores associated with $a$

\For{$s\in[M]\setminus\{s_0\}$}
    \State Use $\mathcal Z_{s_0,s}$ to obtain an
    orthogonal rotation $R_s$ and translation $\bm b_s$
    \State Let
    $\widehat X_s\gets X_sR_s+\bm 1\bm b_s^T$
    be the coordinates aligned with shard $s_0$

    \ForAll{$(u,v)\in V_{s_0}\times V_s$}
        \State Define the MDS matching score as
        $S^{\mathrm{mds}}_{s_0,s}(u,v)
        \gets-\|X_{s_0}[u]-\widehat X_s[v]\|_2$
    \EndFor

    \State Let $\pi_{s_0,s}$ be the one-to-one
    maximum-score assignment that preserves
    $\mathcal Z_{s_0,s}$
    \State Let $v^*\gets\pi_{s_0,s}(a)$ be the entry
    assigned to $a$, or $\bot$ if no entry is accepted

    \If{$v^*\neq\bot$ and
    $S^{\mathrm{mds}}_{s_0,s}(a,v^*)
    \geq\tau_{\mathrm{mds}}$}
        \State Add $(s,v^*)$ to
        $\widehat{\mathcal C}_{\mathrm{mds}}(a)$
        \State Add
        $S^{\mathrm{mds}}_{s_0,s}(a,v^*)$
        to $\mathcal L$
    \EndIf
\EndFor

\State Let $T_{\mathrm{mds}}$ be the fixed scalar
aggregation of the scores in $\mathcal L$
\State \Return
$\widehat{\mathcal C}_{\mathrm{mds}}(a)$ and
$T_{\mathrm{mds}}$

\end{algorithmic}
\end{algorithm}
\end{minipage}
\hfill
\begin{minipage}[t]{0.485\textwidth}
\vspace{0pt}
\begin{algorithm}[H]
\caption{Topology-Based Cross-Shard Linkage}
\label{alg:topology-linkage}
\begin{algorithmic}[1]

\Require Server-visible shards
$\mathsf V=\{\mathcal G^{(s)}=(V_s,E_s)\}_{s=1}^{M}$;
fixed seed pairs $\mathcal Z_{s_0,s}$;
reference shard $s_0$ and target entry $a\in V_{s_0}$;
maximum hop count $h_{\max}$ and matching threshold
$\tau_{\mathrm{top}}$

\Ensure Predicted corresponding-entry set
$\widehat{\mathcal C}_{\mathrm{top}}(a)$ and
distinguishing score $T_{\mathrm{top}}$

\For{$s=1,\ldots,M$}
    \ForAll{$u\in V_s$}
        \State Let $\bm h_s(u)$ be the topology
        descriptor containing the HNSW level, node
        degree, neighbor-degree statistics, and
        structural statistics within $h_{\max}$ hops
    \EndFor
    \State Normalize all descriptors $\bm h_s(u)$
    using fixed calibration statistics
\EndFor

\State Initialize
$\widehat{\mathcal C}_{\mathrm{top}}(a)
\gets\{(s_0,a)\}$
\State Let $\mathcal L$ store the accepted matching
scores associated with $a$

\For{$s\in[M]\setminus\{s_0\}$}
    \ForAll{$(u,v)\in V_{s_0}\times V_s$}
        \State Define the initial topology score as
        $S^{\mathrm{top}}_{s_0,s}(u,v)
        \gets-\|\bm h_{s_0}(u)-\bm h_s(v)\|_2$
        \State Increase the score according to the
        fixed number of seed-supported neighboring
        correspondences
    \EndFor

    \State Initialize the correspondence map
    $\pi_{s_0,s}$ with the fixed seeds
    $\mathcal Z_{s_0,s}$

    \Repeat
        \State Find mutually best unmatched pairs
        according to $S^{\mathrm{top}}_{s_0,s}$
        \State Add pairs whose scores exceed
        $\tau_{\mathrm{top}}$ to $\pi_{s_0,s}$
        as pseudo-seeds
    \Until{no new pair is added}

    \State Let $v^*\gets\pi_{s_0,s}(a)$ be the entry
    assigned to $a$, or $\bot$ if no entry is accepted

    \If{$v^*\neq\bot$ and
    $S^{\mathrm{top}}_{s_0,s}(a,v^*)
    \geq\tau_{\mathrm{top}}$}
        \State Add $(s,v^*)$ to
        $\widehat{\mathcal C}_{\mathrm{top}}(a)$
        \State Add
        $S^{\mathrm{top}}_{s_0,s}(a,v^*)$
        to $\mathcal L$
    \EndIf
\EndFor

\State Let $T_{\mathrm{top}}$ be the fixed scalar
aggregation of the scores in $\mathcal L$
\State \Return
$\widehat{\mathcal C}_{\mathrm{top}}(a)$ and
$T_{\mathrm{top}}$

\end{algorithmic}
\end{algorithm}
\end{minipage}

\end{figure*}
\clearpage

\begin{figure*}[p]
\centering

\begin{minipage}[t]{0.485\textwidth}
\vspace{0pt}
\begin{algorithm}[H]
\caption{Graph-Signal-Based Cross-Shard Linkage}
\label{alg:gsgm-linkage}
\begin{algorithmic}[1]

\Require Server-visible shards
$\mathsf V=\{\mathcal G^{(s)}=(V_s,E_s)\}_{s=1}^{M}$;
perturbed code $\widetilde{\bm v}_{u,s}$ of each
$u\in V_s$; fixed seed pairs $\mathcal Z_{s_0,s}$;
reference shard $s_0$ and target entry $a\in V_{s_0}$;
filter depth $H$, filter coefficient $\alpha$, and
matching threshold $\tau_{\mathrm{gsgm}}$

\Ensure Predicted corresponding-entry set
$\widehat{\mathcal C}_{\mathrm{gsgm}}(a)$ and
distinguishing score $T_{\mathrm{gsgm}}$

\For{$s=1,\ldots,M$}
    \State Let $A_s$ be the adjacency matrix of
    $\mathcal G^{(s)}$
    \State Let $\widehat A_s$ be the degree-normalized
    adjacency matrix derived from $A_s$
    \State Let $F_s^{(0)}$ be the initial node-signal
    matrix obtained from fixed summaries of the
    perturbed binary embeddings

    \For{$h=1,\ldots,H$}
        \State Compute the signal matrix after the
        $h$-th filtering step as
        $F_s^{(h)}\gets
        \alpha F_s^{(0)}
        +(1-\alpha)\widehat A_sF_s^{(h-1)}$
    \EndFor

    \ForAll{$u\in V_s$}
        \State Let $\bm g_s(u)$ be the graph-signal
        descriptor obtained by concatenating
        $F_s^{(0)}[u],\ldots,F_s^{(H)}[u]$
    \EndFor
\EndFor

\State Initialize
$\widehat{\mathcal C}_{\mathrm{gsgm}}(a)
\gets\{(s_0,a)\}$
\State Let $\mathcal L$ store the accepted matching
scores associated with $a$

\For{$s\in[M]\setminus\{s_0\}$}
    \State Use $\mathcal Z_{s_0,s}$ to align
    $\bm g_s$ with the graph-signal descriptor space
    of shard $s_0$
    \State Denote the aligned descriptor by
    $\widehat{\bm g}_s$

    \ForAll{$(u,v)\in V_{s_0}\times V_s$}
        \State Define the graph-signal score as
        $S^{\mathrm{gsgm}}_{s_0,s}(u,v)
        \gets
        -\|\bm g_{s_0}(u)-\widehat{\bm g}_s(v)\|_2$
    \EndFor

    \State Let $\pi_{s_0,s}$ be the one-to-one
    maximum-score assignment that preserves
    $\mathcal Z_{s_0,s}$
    \State Let $v^*\gets\pi_{s_0,s}(a)$ be the entry
    assigned to $a$, or $\bot$ if no entry is accepted

    \If{$v^*\neq\bot$ and
    $S^{\mathrm{gsgm}}_{s_0,s}(a,v^*)
    \geq\tau_{\mathrm{gsgm}}$}
        \State Add $(s,v^*)$ to
        $\widehat{\mathcal C}_{\mathrm{gsgm}}(a)$
        \State Add
        $S^{\mathrm{gsgm}}_{s_0,s}(a,v^*)$
        to $\mathcal L$
    \EndIf
\EndFor

\State Let $T_{\mathrm{gsgm}}$ be the fixed scalar
aggregation of the scores in $\mathcal L$
\State \Return
$\widehat{\mathcal C}_{\mathrm{gsgm}}(a)$ and
$T_{\mathrm{gsgm}}$

\end{algorithmic}
\end{algorithm}
\end{minipage}
\hfill
\begin{minipage}[t]{0.485\textwidth}
\vspace{0pt}
\begin{algorithm}[H]
\caption{Cycle-Consistent Joint Cross-Shard Linkage}
\label{alg:joint-linkage}
\begin{algorithmic}[1]

\Require Pairwise scores
$S^{\mathcal A}_{s,r}(u,v)$ obtained by applying the
scoring stages of MDS, topology, and GSGM to every
shard pair, where
$\mathcal A\in
\{\mathrm{mds},\mathrm{top},\mathrm{gsgm}\}$;
calibration mean $\mu_{\mathcal A}$, standard deviation
$\sigma_{\mathcal A}$, and fixed weight
$\omega_{\mathcal A}$ for each method;
fixed seed pairs $\mathcal Z_{s,r}$;
reference shard $s_0$ and target entry $a$;
thresholds $\tau_{\mathrm{joint}}$ and
$\tau_{\mathrm{cyc}}$; maximum iteration count
$I_{\max}$

\Ensure Jointly predicted corresponding-entry set
$\widehat{\mathcal C}_{\mathrm{joint}}(a)$ and
distinguishing score $T_{\mathrm{joint}}$

\ForAll{shard pairs $(s,r)$}
    \ForAll{candidate pairs $(u,v)\in V_s\times V_r$}
        \ForAll{$\mathcal A\in
        \{\mathrm{mds},\mathrm{top},\mathrm{gsgm}\}$}
            \State Let
            $\overline S^{\mathcal A}_{s,r}(u,v)
            \gets
            \frac{S^{\mathcal A}_{s,r}(u,v)
            -\mu_{\mathcal A}}
            {\sigma_{\mathcal A}}$
            be the standardized score of method
            $\mathcal A$
        \EndFor

        \State Define the combined score as
        $S^{\mathrm{joint}}_{s,r}(u,v)
        \gets
        \sum_{\mathcal A}
        \omega_{\mathcal A}
        \overline S^{\mathcal A}_{s,r}(u,v)$
    \EndFor

    \State Let $\pi_{s,r}$ be the one-to-one
    maximum-score assignment that preserves
    $\mathcal Z_{s,r}$
\EndFor

\For{$i=1,\ldots,I_{\max}$}
    \State Let $\mathsf{changed}$ indicate whether a
    new pseudo-seed is added in this iteration
    \State $\mathsf{changed}\gets\mathsf{false}$

    \ForAll{predicted pairs $v=\pi_{s,r}(u)$}
        \State Let $\rho_{s,r}(u,v)$ be the fraction of
        available third shards $g$ for which the
        indirect mapping satisfies
        $\pi_{g,r}(\pi_{s,g}(u))=v$

        \If{$S^{\mathrm{joint}}_{s,r}(u,v)
        \geq\tau_{\mathrm{joint}}$ and
        $\rho_{s,r}(u,v)\geq\tau_{\mathrm{cyc}}$}
            \If{$(u,v)$ is not already a seed}
                \State Add $(u,v)$ as a pseudo-seed
                \State $\mathsf{changed}
                \gets\mathsf{true}$
            \EndIf
        \EndIf
    \EndFor

    \If{not $\mathsf{changed}$}
        \State \textbf{break}
    \EndIf

    \State Recompute each affected assignment
    $\pi_{s,r}$ using the updated seeds
\EndFor

\State Initialize
$\widehat{\mathcal C}_{\mathrm{joint}}(a)
\gets\{(s_0,a)\}$
\State Let $\mathcal L$ store the accepted joint
matching scores associated with $a$

\For{$s\in[M]\setminus\{s_0\}$}
    \State Let $v^*\gets\pi_{s_0,s}(a)$ be the entry
    jointly assigned to $a$

    \If{$v^*\neq\bot$,
    $S^{\mathrm{joint}}_{s_0,s}(a,v^*)
    \geq\tau_{\mathrm{joint}}$, and
    $\rho_{s_0,s}(a,v^*)\geq\tau_{\mathrm{cyc}}$}
        \State Add $(s,v^*)$ to
        $\widehat{\mathcal C}_{\mathrm{joint}}(a)$
        \State Add
        $S^{\mathrm{joint}}_{s_0,s}(a,v^*)$
        to $\mathcal L$
    \EndIf
\EndFor

\State Let $T_{\mathrm{joint}}$ be the fixed scalar
aggregation of the scores in $\mathcal L$
\State \Return
$\widehat{\mathcal C}_{\mathrm{joint}}(a)$ and
$T_{\mathrm{joint}}$

\end{algorithmic}
\end{algorithm}
\end{minipage}

\end{figure*}
\clearpage

\end{document}